\documentclass[12pt, a4-paper]{amsart}
\usepackage{a4wide,color}

\usepackage{bbm,dsfont}
\usepackage[english]{babel}
\usepackage{amsmath, amsfonts, amsthm, amssymb}
\usepackage{mathrsfs, mathtools}
\usepackage{booktabs}
\usepackage[ruled,vlined]{algorithm2e} 
\usepackage{algpseudocode}
\usepackage{cases}
\usepackage{dblfloatfix}
\usepackage{multirow}
\usepackage{comment}
\usepackage{subfigure, float}

\mathtoolsset{showonlyrefs}

\usepackage{hyperref, url, natbib}
\hypersetup{
  colorlinks   = true,  
  urlcolor     = blue,  
  linkcolor    = blue,  
  citecolor    = red    
}

\usepackage[foot]{amsaddr}

\newcommand{\tr}{\operatorname{tr}}

\newcommand{\cO}{\mathcal{O}}

\newcommand{\cX}{\mathcal{X}}
\newcommand{\cY}{\mathcal{Y}}
\newcommand{\cZ}{\mathcal{Z}}
\newcommand{\cU}{\mathcal{U}}
\newcommand{\R}{\mathbb{R}}
\newcommand{\C}{\mathbb{C}}

\newcommand{\mf}[1]{\mathfrak{#1}}

\numberwithin{equation}{section}

\newcommand{\kl}{\pl \le \pl}
\newcommand{\gl}{\pl \ge \pl}

\newcommand{\lel}{\pl = \pl}

\newcommand{\rz}{{\mathbb R}}
\newcommand{\zz}{{\mathbb Z}}

\newcommand{\cz}{{\mathbb C}}

\newcommand{\ten}{\otimes}

\newcommand{\pl}{\hspace{.1cm}}
\newcommand{\pll}{\hspace{.3cm}}

\newcommand{\ran}{\rangle}
\newcommand{\lan}{\langle}

\newcommand{\al}{\alpha}

\newcommand{\La}{\Lambda}
\newcommand{\la}{\lambda}
\newcommand{\eps}{\varepsilon}

\newcommand{\A}{{\mathcal A}}

\newcommand{\mz}{{\mathbb M}}

\newcommand{\PP}{{\rm I\! P}}

\newcommand{\si}{\sigma}

\newtheorem{lemma}{Lemma}[section]
\newtheorem{theorem}[lemma]{Theorem}
\newtheorem{cor}[lemma]{Corollary}
\newtheorem{corollary}[lemma]{Corollary}

\newtheorem{rem}[lemma]{Remark}

\newtheorem{prop}[lemma]{Proposition}
\newtheorem{proposition}[lemma]{Proposition}
\newcommand{\re}{\begin{rem}\rm}
  \newcommand{\mar}{\end{rem}}

\newcommand{\qd}{\end{proof}\vspace{0.5ex}}
\newcommand{\prf}{\begin{proof}[\bf Proof:]}
\newtheorem{defi}[lemma]{Definition}
\usepackage{xcolor}

\begin{document}

\title{Universal computation with magic Hamiltonians}

\author{Marius Junge}
\address{Department of Mathematics, University of Illinois at Urbana-Champaign, Urbana, IL 61801, USA}
\author{Jason Pollack}
\address{Department of Electrical Engineering \& Computer Science and Institute for Quantum \& Information Sciences, Syracuse University, New York, USA}
\author{Luke Visser}
\address{Department of Applied Mathematics, Eindhoven University of Technology, Eindhoven, The Netherlands}
\email{lukevisser99@gmail.com}

\thanks{MJ is partially supported by NSF DMS 2247114 and NSF DMS 2556195. LV is financially supported by Netherlands Organisation for Scientific Research (NWO) under Grant NGF.1582.22.009.}

\maketitle

\section*{Abstract} \label{ch:abstract}
In the conventional theory of quantum computation, universality is discussed in terms of properties of unitary gate sets. In many experimental setups, however, we instead have access to a parametrized set of Hamiltonians, which can be exponentiated for any desired time. Accordingly, we formulate a continuous analogue of universality, where adding one expensive Hamiltonian control to a cheap set of local control operators allows for universal computation. We find several concrete sets of 1- or 2-local Pauli generators $S$ and $n$-qubit ``magic" Hamiltonians $H$ that together generate the full Lie algebra $\mf{su}(2^n)$, with a special focus on the Ising Hamiltonian $J$. We propose a method to directly implement any unitary by exponentiating these generator sets, and give bounds on the needed time of application of the magic Hamiltonian $H$. We find that the number of applications of $H$ needed to approximate any unitary up to accuracy $\epsilon$ scales exponentially in the number of qubits as expected from the standard application of the Solovay-Kitaev theorem. 
We also find phase transitions in the size of the generated Lie algebras as the parameters of the Ising Hamiltonian are varied. We find an algorithmic application of the Chow/Rashevskii technique to implement unitaries corresponding to iterated commutators.

\newpage
\tableofcontents

\section{Introduction} \label{ch:jason_introduction}

Computer scientists often conceptualize quantum computation as gate-based, by analogy to the classical circuit elements that make up electronic processors. In practice, however, in a particular experimental implementation of a quantum computer some gates may be more expensive or difficult to perform. Quantum information theorists often work within a model \cite{Bravyi:2004isx} in which gates implementing unitary operators in the Clifford group are effectively free while other `magic' gates are expensive. In this paper we develop a continuous analogue: a so-called magic Hamiltonian combined with local cheap resource Hamiltonians 
allows for universal computation. Mathematically speaking, a small Lie algebra is enhanced by one additional generator to generate the full Lie algebra of $n$ qubits.
In quantum information theory, combining different Hamiltonians and the geometry of tangent vectors to the unitary group is used in Nielsen complexity \cite{Nielsen:2005mkt} and dynamical Lie algebras \cite{DAlessandro:2024eff,Aguilar2024FullClassificationLie}.

We are motivated by a physical fact: in actual experimental platforms the starting point is typically not a gate but a Hamiltonian, or, more precisely, a set of Hamiltonians with some control parameters such as driving strengths, magnetic fields, etc. Specific gates are then engineered by evolving the system with one of these Hamiltonians for a particular length of time \cite{DeFouquieres2011SecondEngineering, 2021PulseSystems}.
A transmon qubit, for example, is controlled using only a single type of pulse as well as ``virtual'' changes of basis \cite{McKay:2017rej,Kjaergaard202617SuperconductingQubits}. Multi-qubit gates can similarly be implemented by exponentiated Hamiltonians, such as nonlinear inductive elements (e.g.\ SNAILs) \cite{McKinney2023Co-DesignedComputers}, although some platforms such as neutral atoms may also implement some multi-qubit gates (like SWAP operations) in more fundamentally discrete ways \cite{Henriet2020QuantumAtoms}.

To fully realize a quantum computer a universal gate set in the sense of Solovay-Kitaev is needed, in which any unitary operator can be approximated up to some universal threshold by a circuit composed of gates in the universal gate set \cite{KSV, dawson2006SolovayKitaev}. 
An often-used example of a universal gate set is the combination of Clifford gates enhanced by a single $T$ gate (see e.g. \cite{Bravyi:2004isx}). In this example, the Clifford gates are often easy to implement with high fidelity, while the $T$ gate is the bottleneck. One can then draw a parallel back to the control Hamiltonians, where one requires a complete control set to get the same universality. On different quantum platforms, some implementable control gates may be considered less expensive than others and the cost of a circuit counts the number of times an expensive gate is used. 

In this paper, we use this parallel to develop a Hamiltonian resource theory of universal quantum computation. We are particularly interested in ``magic'' Hamiltonians that, when combined with various easy-to-implement Hamiltonian resource sets, allow one to approximate any unitary on the system. 
\begin{defi}[Magic Hamiltonian] Let $S$ be a fixed resource set of $n$-qubit Hamiltonians. A Hamiltonian $H\in \mz_{2^n}$ is called \emph{magic with respect to $S$} if the Lie algebra generated by $iS$ and $iH$ is $\mathfrak{su}(2^n)$.  
\end{defi}

At least in the current era, where different quantum platforms are still competing to take the lead for feasible quantum computation, the characterization of cheap versus expensive Hamiltonians is not carved in stone. In this paper we imagine setups where an Ising-type Hamiltonian \begin{equation}\label{eq:ising_intro}
    J_{\gamma, \lambda} = \sum_j (1+\gamma) X_j X_{j+1} + (1-\gamma)Y_j Y_{j+1} + \lambda \sum_j Z_j
\end{equation}
acting on all qubits in a system is the expensive resource, and one or a few low-weight Pauli strings acting on only a few qubits in the system are cheap resources. In our resource sets, we will often use Pauli strings acting only on a single qubit or two neighbouring qubits. We denote by $\cZ(1), \cX(1)$, (defined formally in Def.~\ref{def:control_sets}) the sets of generators acting on each single qubit as $Z_i$ and $X_i$ respectively. Similarly, $\cZ(2), \cX(2)$ are the sets of generators acting on each two neighbouring qubits as $Z_i Z_{i+1}$ and $X_i X_{i+1}$ respectively. Other choices using collective Hamiltonians are discussed in Hu et al.\ \cite{hu}.

One main contribution of this paper is to present several natural examples of magic Hamiltonians (see section 3 for a more comprehensive list). 
\begin{theorem}\label{sample} The following table, excerpted from Theorem \ref{thm:controlsets}, provides examples and non-examples of magic Hamiltonians: 
    \[ \begin{array}{c|c|c|c|c}
 \text{No.}& \text{local set } S & \text{Dim}(L(S)) & \text{Hamiltonian} & \text{magic} \\  \hline
 1)& \mathcal{X}(2)\cup \mathcal{Z}(1)& 2n^2-n& X_1+Z_{1}Z_2 & \text{yes}  \\ 
  5)&\{X_1,Z_1Z_2\}\cup \mathcal{Z}(1) & n+5 & J_{\gamma,\lambda} &  \text{yes}\\
  8)&\{X_1\}\cup \mathcal{Z}(2)  & n+1 & J_{\gamma,\lambda} & \text{yes}^* \\ 
  \end{array} \] 
Here yes$^*$ is under the condition $\la\neq 0$, $\gamma\notin\{-1,0,1\}$.
\end{theorem}

The first example 1) demonstrates that even a local Hamiltonian can be magic with respect to local resource sets. Example 5) shows that Ising Hamiltonians are magic even with respect to nearly-commutative resources augmented by full control of one qubit. Example 8) is inspired by the work of Araki and Matsui on phase transitions \cite{AM}: small changes of parameters have significant impact on the global system. Whereas the work of Araki and Matsui considers the properties of the infinite-dimensional ground state, our phase transitions concern the ability to perform universal computation (see Appendix \ref{app:operator_algebra} for more discussion of the relation of our work to operator-algebra theory).

Given a particular choice of magic Hamiltonian, another main contribution of our work is to give explicit algorithmic recipes for constructing any desired gate as a sequence of Hamiltonian applications. Indeed, for a set $S$ of Hamiltonians one can generate gates of the form  
 \begin{equation}\label{prodd} 
 u\lel e^{it_1h_1}e^{it_2h2}\cdots e^{it_mh_m} \pl ,\pl t_j\in \rz, h_j\in S
 \end{equation}  
by choosing the control parameters $t_h$ and the order of the Hamiltonians $h_j$. The closure of the set of all such controllable unitaries is called the Lie group generated by the Lie algebra of iterated commutators of the resource set $S$. The work of Chow \cite{chow1940systeme} and Rashevskii \cite{rashevsky1938connecting} 
from the forties showed geometrically that any gate can be obtained by a continuous version of \eqref{prodd} (used in \cite {hu}). Rashevkii's method, also used to prove the related famous Ball-Box theorem, even provides an algorithmic construction of this approximation procedure, which we do not believe is widely known in the quantum computing literature. The key feature is to consider the zigzag
 \[ e^{ith_1}e^{ith_2}e^{-ith_1}e^{ith_2}\lel e^{t^2[h_1,h_2]}+O(t^3) \pl .\]
Here $[a,b]=ab-ba$ is the usual commutator. The proof of Theorem \eqref{sample} provides a quasi-algorithmic method for universal computation:  
 \begin{enumerate}
 \item[1)] Calculate $L(S)$, the span of all iterated commutators $[H_{i_1},[H_{i_2},\cdots [H_{i_{k-1}},H_{i_k}]$ for elements in $S$,
  \item[2)] Show that adding the magic Hamiltonian boosts $L(S)$ to $L(S\cup \{H\})=\mathfrak{su}(2^n)$, the Lie algebra of the unitaries for qubits.
  \item[3)] Apply Rashveskii's zigzagging method to approximate $e^{ta}$ for a basis of elements in $L(S\cup\{H\})$.
  \item[4)] Use Trotterization to approximate every unitary $u=e^{\sum_j \la_j a_j}$.
   \end{enumerate}
   
In our context all resource sets are Pauli strings. Let us recall that Pauli strings are products
 \[ P\lel V_{j_1}\ten \cdots \ten V_{j_m} \pl ,\pl V_j\in \{I,X,Y,Z\} \pl .\] 
The key insights enabling 1) are the orthogonality of Pauli strings and the commutator formula (see also \cite{Aguilar2024FullClassificationLie})  
\begin{equation}\label{pq}
  [P,Q] \lel 2 \delta_{[P,Q]=0} PQ \pl .
 \end{equation} Finding the right order of iterated commutators is then a task for a classical computer. Step 2) is essentially a repetition of step 1 if $H$ is a linear combination of Pauli strings.   
Rashevekii's method 3) and Trotterization 4) are essentially algorithmic. Moreover, the cost of Trotterization is well-studied \cite{Childs:2019hts}. It therefore appears that every individual step, although non-trivial, should be feasible to allow explicit algorithmic constructions of approximating unitaries. (Indeed, we write a formal algorithm for this procedure in Appendix \ref{app:alg-computation}.)  

In the examples in Theorem \ref{sample} (and Theorem \ref{thm:controlsets}) the dimensions of $L(S)$, the Lie algebra generated by the resource sets alone, are polynomial in the number of qubits. (In most of our examples we can even explicitly construct the Lie algebras $L(S)$ and their presentations as Pauli matrices, going slightly beyond the classification results of \cite{Aguilar2024FullClassificationLie}.) Since the dimension of the full Lie algebra of all Pauli strings is $4^{n}-1$, the magic Hamiltonian has to enable a gigantic boost, enabling the construction of many more gates. Our final main contribution is to show that most of these gates are very expensive. To be more formal, define (anticipating Definition \ref{def:alt_approx_length}) the alternating approximation length $l_H(\eps|S)$ as the smallest $m$ such that  any arbitrary $u$ can be approximated up to $\eps$ by a sequence alternating $m$ times between evolution with a resource and evolution with the magic Hamiltonian: 
 \[ \|u-u_1\cdots u_{2m+1}\|<\eps \pl ,\pl u_{2j+1}\in G(S) \pl ,\pl u_{2j}=e^{it_jH} \pl.\]
This alternating approximation length can be computed for arbitrary unitaries in the Lie group $G(S)=\exp(L(S))$ given by the resources and control in $H$.       

\begin{theorem} All the magic Hamiltonians in Theorem \ref{sample} and section 3 have exponential approximation length (in the number of qubits).
\end{theorem}

The Solovay-Kitaev theorem guarantees that the dependence on $\eps$ in the definition of the alternating length is not really essential below a certain threshold. The quasi-algorithm above also shows that a finite number of ``time stamps'', i.e.\ durations by which the Hamiltonians can be evolved, are sufficient for the construction of universal gate sets with the help of finite nets in $G(S)$.

The remainder of this paper is organized as follows. In section \ref{sec:math_prelim}, we study abstract properties of Lie groups and Lie algebras generated by Pauli strings. Then in section \ref{sec:generated-Lie-Alg} we explicitly calculate the Lie algebras generated by specific resource sets and find corresponding magic Hamiltonians which, together with their resource sets, enable all possible unitaries to be created. These results are collected in Theorem~\ref{themagic}. Armed with these results on generated Lie algebras, we pass to the corresponding Lie groups and their implementation in section \ref{sec:algorithmic_implementation} using Trotterization and the Chow-Rashevskii theory. One important component of this implementation is the number of commutators needed to generate an arbitrary Pauli string $P$, which we calculate in Theorem~\ref{thm:commutators_algorithmically}. These implementations using Trotterization and Chow-Rashevskii use alternating applications of the given generators to create a given unitary. In section \ref{sec:approx_length}, we investigate how often one has to alternate between the magic Hamiltonian and elements of the local generator set to approximate any unitary with precision $\epsilon$, where we find in Theorem~\ref{thm:exponential_m} that the number of Hamiltonian applications grows exponentially in the number of qubits for all our examples based on Szarek's covering number results \cite{Szarek1998MetricSpaces}. We conclude in section \ref{sec:conclusion}.

\paragraph{\textit{Acknowledgements}} Part of this research was performed while the authors were visiting the Institute for Pure and Applied Mathematics (IPAM), which is supported by the National Science Foundation (Grant Nos. DMS-1925919 and DMS-2422832). The authors acknowledge discussions with Ethan Arnault, Oliver Tse, David Garcia-Perez, Aaram Harrow, and Hong-Ye Hu.

\section{Lie-algebra preliminaries}\label{sec:math_prelim}

On a quantum computer with $n$ qubits, the possible gates are unitary operators, elements of $\cU(2^n)$. The control operators, which can be exponentiated to give the implementable gates, live in the Lie algebra, which consists of all Hermitian matrices in $\C^{2^n \times 2^n}$, with a binary operation, the Lie bracket, given by the commutator
\begin{equation}
    [A,B] = AB - BA.
\end{equation}

In quantum computing, the set of Pauli matrices plays a fundamental role. The Pauli matrices are given by 
\begin{equation}
\begin{aligned}
        I & = \begin{pmatrix}
        1 & 0 \\0 & 1
    \end{pmatrix} & X = \sigma_x & = \begin{pmatrix}
        0 & 1 \\1 & 0
    \end{pmatrix} \\
    Y = \sigma_y & = \begin{pmatrix}
        0 & -i \\i & 0
    \end{pmatrix} & Z = \sigma_z  & = \begin{pmatrix}
        1 & 0 \\0 & -1
    \end{pmatrix}.
\end{aligned}
\end{equation}
The commutator between two Pauli matrices is given by 
\begin{equation}
    [\sigma_i, \sigma_j] = 2 i \sigma_k,
\end{equation}
for $i, j, k$ a cyclic permutation of $x, y, z$, and zero otherwise. The matrices $\{I, i\sigma_x, i\sigma_y, i\sigma_z\}$ span the Lie algebra acting on a single qubit. For $n$ qubits, the Lie algebra is spanned by Pauli strings, the tensor products of Pauli matrices $\alpha_j \in \{I,X,Y,Z\}$ acting on qubit $j$ as
\begin{equation}
    V\lel  \alpha_1 \otimes \alpha_2 \otimes \cdots \otimes \alpha_n.
\end{equation}
The well-known equality (e.g., \cite{Aguilar2024FullClassificationLie})
\begin{equation}\label{comP}
  [V,W] \lel \begin{cases} 2VW & [V,W]\neq 0 \\ 
                              0 & [V,W] \lel 0 \end{cases} 
                              \end{equation}
 for arbitrary Pauli strings $V,W$ as above will play a crucial role in this paper, allowing us to use group-theoretic arguments. 
 
 In many cases we will take our starting resource set to be a subset $S\lel \{V_1,...,V_m\}$ of Pauli strings. By finding the sets of all iterated Lie brackets, we can determine the corresponding Lie algebra. Such a system is called a H\"ormander system \cite{10.1007/BF02392081} (bracket generating) if the iterated commutators generate all non-trivial Pauli strings. In particular, the Lie algebra generated by Pauli strings will have a basis also consisting of Pauli strings. 
 
 In the Lie algebra setting, anti-selfadjoint generators $X^{*}=-X$ are the preferred choice. So formally, we should talk about the Lie algebra generated by $iS$, whereas calculations for groups prefer to work with selfadjoint strings.  We will work with Pauli strings and freely switch between these two perspectives. For readers familiar with the stabilizer formalism, we note that a given set of Pauli strings also generates a discrete subgroup inside the unitary group, and we will investigate the relation between these objects and their representation theory. Nevertheless, the concrete representation in $\mz_{2^n}$ perfectly reflects the linear stricture. Even better, subsets of Pauli strings are linearly independent if and only if they are orthogonal in the Hilbert-Schmidt sense. 

In the following we collect some abstract facts about the groups and Lie algebras generated by Pauli stings. It turns out that maps which send Pauli's to Pauli's and preserve the commutation relation can be extended to homomorphisms of the abstract group generated by the given representation, and even to Lie-algebra and $C^*$-isomorphisms.  For commuting Pauli generators, this is a fairly obvious, often-used fact; we need the non-commuting counterpart. These abstract arguments will be used to help identify Lie groups and Lie algebras: for example, by mapping a set of Pauli strings to a different Pauli strings with the same commutation relations but using fewer qubits.

\begin{defi} \label{def:paulie} Let $\{P_1,...,P_m\}\subset U(2^n)$ be a set of Pauli strings. 
 \begin{enumerate}
 \item[i)] We denote by $ D=D(P_1,...,P_m) \lel \{[P_{i_1},[P_{i_2},\cdots ]]\cdots ]] | 1\le  i_1,...,i_k\le m \}$  
the set of their iterated commutators.  
\item[ii)] We denote by $G(P_1,...,P_m)\subset U(2^n)$ the abstract subgroup they generate. 
\item[iii)] Let $K=\{1,-1\}\subset U(2^n)$. By $G(K,P_1,...,P_m)$ we denote the subgroup generated by $K$ and $P_1...,P_m$. 
  \item[ii)] A map  $\si:\{P_1,...,P_m\}\to \mz_{2^m}$ is said to be \textbf{PauLie} if $\si(P_j)$ is again a Pauli string and
 \[ [P_j,P_k]\lel 0 \pll \Leftrightarrow \pll [\si(P_j),\si(P_k)] \lel 0 \pl .\]
 \end{enumerate}
 \end{defi}

\begin{lemma}\label{lem:PauLie_homomorphism} Let $P_1,..,P_m$'s be Pauli strings. The group 
$G(K,P_1,...,P_m)$ is uniquely specified by the  matrix of coefficients $\eps_{ij}$ in the relation 
 \[ P_iP_j\lel \eps_{ij}P_jP_i \pl .\]  
In particular, a  PauLie map induces a group homomorphism.
\end{lemma}

\begin{proof} Let $G_k=G(K,P_1,...,P_k)$ be the group generated by the first $k$ elements. We note that for a fixed $P_{k+1}$  the map 
 \[ \al(Q) \lel P_{k+1}QP_{k+1} \pl \]
is a group homomorphism. Therefore $G_{k+1}=G_k\rtimes_{\al}\zz_2$ can be interpreted 
as  a semi-direct product consisting of pairs $(Q_1,Q_2P_{k+1})$ with $Q_{1,2}\in G_k$ using 
$Q_1P_{k+1}Q_3=Q_1\al(Q_3)P_{k+1}$. Given a PauLie-map $\si$ we can proceed by induction to show that the canonical extension $\si(P_{i_1}\cdots P_{i_l})=\si(P_{i_1})\cdots \si(P_{i_l})$ is a well-defined group homomorphism. Indeed, since the matrix $\eps_{ij}$ is preserved, we see that 
 \[  \si(P_{k+1})\si(Q)\si(P_{k+1}) \lel \si(\al(Q))   \]
 induces exactly the same automorphism.
 \end{proof}

In the following we use the standard notation  $\cz[G]=\{\sum_g \al_g g| \al_g\in\cz\}$ for the group algebra of a group $G$. Here $g$ serves as a symbol denoting linear independent vector indexed by group elements and we only allow finitely many $\al_g$ to be non-zero. The multiplication rule is determined by bilinearity and the group law.  
\begin{cor}\label{LGG} Let $\si$ be a PauLie map, $\si^{G}$  be the unique extension of $\si^G$  to the group algebra $\cz[G(K,P_1,...P_m)]$, and $\si^{Lie}$ be the unique Lie algebra extension. Then
 \[ \si^{G}|_{D(P_1,...,P_m)} \lel \si^{Lie} \pl .\]
\end{cor} 

\begin{proof} By linear independence and linearity, the group homomorphism $\si^{G}$ extends to $\cz[G]$. Then 
\[ \si^G([P,Q]) \lel 2\delta_{[P,Q]\neq 0}
\si^G(PQ) \lel  2\delta_{[P,Q]\neq 0} \si^G(P)\si^G(Q)
\lel [\si^G(P),\si^G(Q)] \pl . \] 
By applying this relation iteratively and observing that iterated commutators are linear multiples of Pauli strings, we deduce that the restriction of the group homomorphism induces an extension on the Lie algebra which preserves commutators. By the uniqueness assumption, this coincides with the Lie group extension of the PauLie map.
\end{proof}

\begin{lemma}\label{cstar} Let $G\subset U(n)$ be a group and $\pi:G\to U(m)$ be a group homomorphism such that $\tau(\pi(g))=\tau(g)$ holds for the normalized trace $\tau(x)=\frac{tr(x)}{tr(1)}$. Then $\pi$ extends to a $C^*$-homomorphism.
\end{lemma}

\begin{proof} We note that $A=\{\sum_g \al_g g|\al_g \in \cz\}$ is a $^*$-algebra. Then 
 $\pi(\sum_g\al_g g)=\sum_{g}\al_g \pi(g)$. Note that for $x=\sum_g a_gg$ we have  
  \[ 
  \tau(\pi(x)^*\pi(x)) \lel \sum_{g,g'}\bar{\al}_g\al_{g'}\tau(\pi(g^{-1}g') )
  \lel \tau(x^*x) \pl. 
  \] 
This implies that $\pi$ is injective and well-defined. Since we are in the finite-dimensional case, we obtain a $^*$-homomorphism defined on the $C^*$-algebra $A$. \end{proof}  

Our next result reveals the metric properties 
of a PauLie map in terms of the Lie group generated by Pauli generators. 

\begin{cor} Let $P_1,...,P_m$ be a set of Pauli strings and $\si$ a PauLie map with values in $\mz_d$. Then the extension $\si^G$ to the group algebra of $\cz[G(K,P_1,...,P_m)]$ extends to an injective $C^*$-homomorphism and also induces a group homomorphism on the Lie group  which preserves the Finsler metric given by the operator norm (see section \ref{sub:univeral_sets})  and an isometry for the distance given by the operator norm. 
\end{cor} 

\begin{proof} Let $\si^{G}:\cz[G(K,P_1,...,P_m)]\to \mz_d$ be the induced $^*$-algebra homorphism. Since the map $\si$ is PauLie it is also trace preserving and extends to a trace preserving map. By Lemma \ref{cstar}, we deduce that $\si^{G}$ is a $^*$-homomorphism on $A=\cz[G]\subset \mz_n$. Note that the identity map is also PauLie, and hence the concrete embedding of $\cz[G]$ is indeed injective, hence $\si:A\to \mz_d$ is an injective $^*$-homomorphism.  Since we are  in finite dimension, $A$ is also closed.  Recall that injective $^*$-homorphisms are indeed isometries, see \cite{TAk1}. Then we observe that for all $j$, the operator $P_j$ is in the image of  the group algebra. The same is true for $e^{itP_j}$ for all $t$ and iterated products. This means the Lie group generated by $G$ is also in $A$. The Finsler metric on $G$ is given by 
 \[ d_{op}(g,h) \lel \inf \int_0^1 \|\gamma'(t)\|_{op} dt, \]
where the infimum is taken over all piecewise differentiable paths $\gamma(t)\in G$. Then $\gamma'(g)=h(t)\gamma(t)$ is given by a vector $h(t)$ in the Lie algebra. In particular 
 \[ d_{op}(\si^G(g),\si^G(h)) \kl \int_0^1 \|\si^G(h(t))\|_{op} dt \pl. \]
Taking the infimum shows that $d(\si^G(g),\si^G(h))\kl d(g,h)$. We have equality by applying this to the inverse of $\si$, again a PauLie map.   
\end{proof}

The corollary shows that PauLie maps can automatically extended to metric-preserving Lie group homomorphisms $\pi=\si^G$ such that the canonical  differential $d\pi_{1}=\si^{Lie}$ is the unique Lie algebra extension. As shown by the canonical tensor-product representation $\pi(g)=g^{\ten_n}$, metric preservation is by no means automatic. The uniqueness of $\si^{Lie}$ also shows that the number of iterated commutators required to build  iterated Pauli strings is a PauLie invariant. Let us formalize this length function:
\begin{defi}\label{def:length_function}
    Define the length function $l(R)$ as the minimal number of commutators with elements from a generator set $G$ that is required to make $R$. If we define $G_k$ iteratively, with $G_1 = G$ and $G_{k+1}$ as 
    \begin{equation}
        G_{k+1} = G_k \cup \{[P_i, P_j] ~:~ P_i \in G_1,~ P_j \in G_k \},
    \end{equation}
    we can formally define $l(R)$ as
    \begin{equation}
        l(R)=\inf \big\{ k ~:~ R \in \operatorname{span}(G_k)\}.
    \end{equation}
\end{defi}
The next observation extends the commutation relation of \eqref{comP} to three Pauli strings. We implicitly used such an iterative formula in Corollary \ref{LGG}. 
\begin{lemma}\label{lem:paulis_closed} Let $R,S,T$ be Pauli strings. Then 
 \[ [R,[S,T]] \lel 4 \delta_{[S,T]\neq 0}(\delta_{[R,S]\neq 0}\delta_{[R,T]=0}+
 \delta_{[R,S]= 0}\delta_{[R,T]\neq 0}) RST \pl .  
  \]
\end{lemma} 
\begin{proof} Recall the Jacobi identity
 \begin{align*}
  [R,[S,T]] &= [R,ST]-[R,TS] \lel [R,S]T+S[R,T]-[R,T]S-T[R,S] \\
  &= 
   [[R,S],T]+[S,[R,T]] \pl .
  \end{align*}
Because commutators of Paulis obey \eqref{comP} we have only to consider the cases where each commutator is resepctively zero or nonzero. If $[R,S]=0=[R,T]$, we find $0$. If $[R,S]\neq 0=[R,T]$ we find
 \[ 2[RS,T] \lel 2[R,T]S+2R[S,T] \lel 4\delta_{[S,T]\neq 0} RST \pl .\] 
If $[R,S]=0\neq [R,T]$, we get
 \[ 2[S,RT] \lel 2[S,R]T+2R[S,T] \lel  4 \delta_{[S,T]\neq 0} RST \pl .\] 
If $[R,S]\neq 0\neq [R,T]$ and $[S,T]\neq 0$, we find 
  \begin{align*}
   2([RS,T]+ [S,RT])&= 2([R,T]S+R[S,T]+[S,R]T+R[S,T] \\
   &=  4 (RTS+RST+SRT+RST) \lel 0 \pl .
   \end{align*}
Collecting all cases yields the asertion.   
\end{proof}

An easy application of the Jacobi identity shows that this length function is subadditive. This is a very useful tool for the concrete calculations we will carry out in the next section. 

\begin{lemma}\label{additive} $l([P,Q])\kl l(P)+l(Q)$.
\end{lemma}

\begin{proof} For a generator $P_{i_j}$ we obviously have
  \[ l([P_{i_j},Q])\kl 1+ l(Q) \pl. \]
This is the needed induction step in proving $l(P)\le k$ implies $l([P,Q])\kl l(P)+l(Q)$. Thus we may assume $P=[P_{i_k},\hat{P}]$ with $\hat{P}$ of length $k-1$. Then the Jacobi idenity tells us that 
 \[ [P,Q] \lel [[P_{i_m},\hat{P}],Q] 
 \lel [P_{i_m},[\hat{P},Q]]+ [\hat{P},[Q,P_{i_m}] ]\pl. \] 
By the induction hypothesis, $l([\hat{P},Q])\kl k-1+l(Q)$, and hence we find an upper estimate $k+l(Q)$ for the first term. Of course, $l([P_{i_m},Q])\kl l(Q)+1$, and again the induction hypothesis applies for the second term (because we ignore scalar coefficients). However, $[P,Q]=2PQ$ is a non-trivial multiple of a Pauli string and this also true for the two terms on the right. By the orthogonality of Pauli strings \eqref{comP}, $[P,Q]$ has to coincide with one of them, up to scalar multiples. But both of $\{P,Q\}$ have length $\le k+l(Q)$. 
\end{proof}

We will see in the next section that PauLie maps can be used to transform existing collections of Pauli strings to new collections using fewer qubits, which are then easier to analyze. We will also exploit the metric-preserving property of PauLie maps for our covering number estimates in connection with covering numbers in Section \ref{sec:approx_length}.  

\section{Concrete examples of generated Lie algebras} \label{sec:generated-Lie-Alg}

In this section, we explicitly calculate the Lie algebras generated by specific resource sets and identify corresponding magic Hamiltonians. We begin in Section \ref{sec:Pauli_Lie} by calculating Lie algebras generated by specific generator sets $S$ of Pauli strings and make an explicit connection to the corresponding operations on a physical quantum computer. In section~\ref{sec:Ising_Lie} we then identify examples of magic Hamiltonians with respect to the resource sets $S$. Most of the magic Hamiltonians we give are parametrized Ising Hamiltonians 
\begin{equation}
    J_{\gamma, \lambda} = \sum_j (1+\gamma) X_j X_{j+1} + (1-\gamma)Y_j Y_{j+1} + \lambda \sum_j Z_j
\end{equation}
for various values of $\gamma$ and $\lambda$. We compile all our results on magic Hamiltonians in Theorem \ref{themagic}.

For all of our results, the resulting Lie algebras can be classified by using generator sets consisting of Pauli strings. This connects it to earlier work on classifying any Lie algebra generated by Pauli strings \cite{Aguilar2024FullClassificationLie}, which we summarize in Remark \ref{rem:Eisert_classification}. In contrast to the abstract characterization from \cite{Aguilar2024FullClassificationLie}, here we insist on specific implementations on $n$ qubits, using the results from Section 2 as mathematical tools.

Most of our resource sets will consist of specific Pauli operations on each single qubit or pair neighboring qubits, as well as perhaps additional operators on a few specific qubits. To make this precise, we define the sets of generators $\cZ(1), \cX(1), \cZ(2), \cX(2)$.
\begin{defi}[Control sets]\label{def:control_sets}
    For a given number of qubits $n$, we define the control sets $\cX(1)$ and $\cX(2)$ as
    \begin{align*}
        \cX(1) &= \{X_i ~:~ i \in [n]\} \\
        \cX(2) &= \{X_i X_{i+1} ~:~ i \in [n] \},
    \end{align*}
    where $X_i$ is the operator
    \begin{equation}
        X_i = \underbrace{I \otimes I \otimes \cdots \otimes I}_{i-1 \text{ times} }\otimes X \otimes \underbrace{I \otimes \cdots \otimes I}_{n -i \text{ times}}.
    \end{equation}
    $\cY(1), \cZ(1), \cY(2), \cZ(2)$ are defined analogously.
\end{defi}
Note that all generators in Definition \ref{def:control_sets} are geometrically $1$- or $2$-local, and this remains true when they are exponentiated.

Our proofs will often proceed by presenting explicit chains of commutators to show that some element is contained in a certain Lie algebra. Pauli strings have the interesting property that when three Pauli strings $A, B$, and $C$ satisfy the property that $[A,B]= 2C$, then they also satisfy $[A,C] = 2B$. The essential feature is thus that $B$ and $C$ are transformed into each other (up to a multiplicative factor) by commutation with $A$, and we will introduce an equivalence relation and notation to highlight this fact.
\begin{defi} For Pauli strings $A\simeq B$ is defined as $A=cB$ for $c\neq 0$. 
\end{defi} 
This definition will avoid unnecessary scalar factors in longer commutator relations, a very special feature of Pauli strings. When a triple of Pauli strings obey $[A,B]\simeq C$ we define a reversible operation denoted by a double arrow
\begin{equation}
    B\stackrel{A}{\longleftrightarrow}C,
\end{equation} 
which we read as ``commutation with $A$ interchanges $B$ and $C$.'' 
When we have not single Pauli strings but linear combinations, commuting with a single other Pauli string is no longer a reversible operation. In that case, we write a single directional arrow
\begin{equation}
    B\stackrel{A}{\longrightarrow}C,
\end{equation} 
which we read as ``commutation with $A$ takes $B$ to $C$,'' for the operation $[A, B] = \alpha C$ for some constant $\alpha$. We have found that notating the chains of commutators in this way reduces the complexity of calculating iterated commutators and removes the need to track the irrelevant constant factors. Furthermore, we sometimes treat the Pauli strings as words and freely omit the indices on Pauli strings, mainly when the indices are irrelevant due to translation invariance.

Applying these tools to various generator sets, we will characterize the number of Pauli $X, Y$, and $Z$ operators in any generated Pauli string $P$. We will denote by $k_x(P), k_y(P)$, and $k_z(P)$ the number of $X, Y$, and $Z$ operators in a given Pauli string $P$.

\subsection{Pauli-generated Lie algebras from short strings}\label{sec:Pauli_Lie}

We begin with some simple Lie algebras generated by Pauli strings and build up to control sets that generate the full Lie algebra. We note that all of the Lie algebras in this subsection have already been classified in \cite{Aguilar2024FullClassificationLie} using different methods. Our contribution here is the explicit representations of the Lie algebras on a quantum computer, i.e.\ as operators acting on qubits. This discussion will also build some understanding of the underlying mathematics that we will use in Section \ref{sec:Ising_Lie}.

\begin{proposition}\label{prop:single-double}
    The Lie algebra generated by $\cX(2)$ and $\cZ(1)$ is spanned by
    \begin{equation}
        \langle \cX(2) \cup \cZ(1) \rangle = \cZ(1) \cup \{\alpha_i Z^{\otimes m} \beta_{i+m+1} ~:~ \alpha, \beta \in \{X,Y\},~ m \in [n-2] \},
    \end{equation}
    which consists of all the individual Pauli $Z$ generators, and the Pauli strings starting and ending with a Pauli $X$ or $Y$ matrix with only Pauli $Z$ matrices in between. The Lie algebra is $\mf{so}(2n)$ and has a dimension of $2n^2-n$.
\end{proposition}
Here and in the following $\al_iZ^{\ten m}\beta_{j+m+1}$ will stand for $\al_iZ_{i+1}\cdots Z_{i+m}\beta_{j+m+1}$. 
\begin{proof}
We first show that $\alpha_i Z^{\otimes m} \beta_{i+m+1} \in \langle \cX(2) \cup \cZ(1)\rangle$, and then we show that this set is closed under the commutator. 

We start the calculations with $X_i X_{i+1}$ and commute with other elements of the generator set to find
\begin{equation}
    X_i X_{i+1} \stackrel{Z_{i+1}}{\longleftrightarrow}X_i Y_{i+1} \stackrel{X_{i+1}X_{i+2}}{\longleftrightarrow} X_i Z_{i+1} X_{i+2} \stackrel{Z_{i+2}}{\longleftrightarrow} \cdots \stackrel{X_{i+m} X_{i+m+1}}{\longleftrightarrow} X_i Z^{\otimes m} X_{i+m+1}.
\end{equation}
The final term we can commute with $Z_i$ and/or $Z_{i+m+1}$ depending on whether $\alpha_i$ and $\beta_{i+m+1}$ are $X$ or $Y$ operators. Thus, we find $\alpha_i Z_{i+1}\cdots Z_{i+m_1-1}\beta_{i+m_1} \in \langle \cX(2) \cup \cZ(1)\rangle$. 

Next, we show that the vector space spanned by these chains and the $\cZ(1)$ is closed under the Lie bracket. First, we notice that the commutator of two chains will result in either a new chain of the same form or an individual $Z_i$. Secondly, the individual chains preserve their form under commutation with $Z_i$. Finally, the commutator of a $Z_i$ with a $Z_j$ will always be 0. Together, this means that the proposed vector space is closed under the commutator and that the generated sub-Lie algebra is thus given by 
\begin{equation}
    \langle \cX(2) \cup \cZ(1) \rangle = \{Z_j|j \in [n]\} \cup \{\alpha_iZ^{\ten_m}\beta_{i+m+2}| \alpha, \beta \in \{X,Y\},~ m \in [n-2] \}.
\end{equation}

Using the classification method for Lie algebras generated by Pauli strings given in \cite{Aguilar2024FullClassificationLie}, we find that the Lie algebra is $\mf{so}(2n)$ with dimension $2n^2-n$.

We can also find this dimension by explicitly counting all possible Pauli strings to get
\begin{equation}
    n+\sum_{k=2}^{n}4(n-k+1) = 2n^2-n. \qedhere
\end{equation}
\end{proof}

We will explain the well-known connection of this Lie algebra with the Jordan-Wigner transform in Section \ref{sec:approx_length}.

Our second explicit example takes $\cY(2)$ instead of $\cZ(1)$, which results in Pauli strings of the same shape, but with stricter conditions on the endpoints.
\begin{proposition}\label{prop:double-double}
    The Lie algebra generated by $\cX(2)$ and $\cY(2)$ is spanned by
    \begin{equation}
            \langle \cX(2) \cup \cY(2) \rangle = \left\{\alpha_i Z^{\otimes m} \beta_{i+m+1} ~:~ 
                (\alpha_i, \beta_{i+m+1}) \in \begin{cases}
                \{XX, YY\} ~\text{ if $m$ is even,} \\ 
                \{XY, YX\} ~\text{ if $m$ is odd}
            \end{cases}
            \right\},
    \end{equation}
    which consists of all the Pauli strings starting and ending with a Pauli $X$ or $Y$ matrix with only Pauli $Z$ matrices in between. The starting and ending Pauli strings must be the same if there are an even number of Pauli $Z$ in between, and different if the number of Pauli $Z$ in between is odd. The resulting Lie algebra is $\oplus_{i=1}^2 \mf{so}(n)$ \cite{Aguilar2024FullClassificationLie} and has dimension $n^2-n$.

    Furthermore, we find that
    \begin{equation}
        \langle \cX(2) \cup \cY(2) \cup \{Z_1\} \rangle = \langle \cX(2) \cup \cZ(1)\rangle,
    \end{equation}
    which is the Lie algebra discussed in Proposition~\ref{prop:single-double}.
\end{proposition}
\begin{proof}
We can use the results of Section \ref{sec:math_prelim} by observing that we can group the generators into two subgroups, where each element of one group commutes with all the elements of the other group. For the case that $n$ is even, this grouping is given by
\begin{equation}\label{2.20}
\begin{aligned}
    S_1 &= \{X_1 X_2, Y_2 Y_3, \cdots, Y_{n-2} Y_{n-1}, X_{n-1} X_n\},  \\
    S_2 &= \{Y_1 Y_2, X_2 X_3, \cdots, X_{n-2} X_{n-1}, Y_{n-1} Y_n\}.
\end{aligned}
\end{equation}
These generate the same Lie algebra as the generators
\begin{equation}
\begin{aligned}
    T_1 &= \{X_1, Y_1 Y_2, X_2, Y_2 Y_3, \cdots, Y_{n/2-1} Y_{n/2}, X_{n/2} \} \\
    T_2 &= \{Y_{n/2+1}, X_{n/2+1}X_{n/2+2}, Y_{n/2+2}, \cdots  X_{n-1} X_n, Y_n \}.
\end{aligned}
\end{equation}
The map $\sigma$ from $S_1, S_2$ to $T_1, T_2$ is a PauLie maps (Definition \ref{def:paulie}). As discussed in Section 2, because this mapping preserves the commutation relations between the elements, it generates the same Lie algebra and we do not have to construct an explicit map. We nevertheless provide such a map and comment on its properties in App.~\ref{app:explicit_map_S_T}.

We can now calculate the algebra generated by $S_1$ and $S_2$ by investigating the algebra generated by $T_1$ and $T_2$. First, note that $T_1$ and $T_2$ each generate $\mf{so}(n)$, the same algebra as given in Proposition~\ref{prop:single-double} for $n/2$. Since $T_1$ and $T_2$ commute, the Lie algebra generated by both is the direct sum of the individual algebras and thus given by $\oplus_{i=1}^2 \mf{so}(n)$.

However, to find the representation of this algebra that corresponds to the original algebra, one has to invert the PauLie map. We do this in appendix~\ref{app:explicit_map_S_T}, but here we find these elements via a direct calculation.

For this direct calculation, we first show that $\alpha_i Z^{\otimes m} \beta_{i+m+1} \in \langle \cX(2) \cup \cY(2) \rangle$. Assume w.l.o.g. by symmetry that $\alpha_i = X$. Applying the following chain of commutators for $m$ even gives
\begin{equation}
    X_i X_{i+1} \stackrel{Y_{i+1} Y_{i+2}}{\longleftrightarrow} X_i Z_{i+1} Y_{i+2} \stackrel{X_{i+2} X_{i+3}}{\longleftrightarrow} X_i Z_{i+1} Z_{i+2} X_{i+3} \stackrel{Y_{i+3} Y_{i+4}}{\longleftrightarrow} \cdots \stackrel{X_{i+m} X_{i+m+1}}{\longleftrightarrow} X_i Z^{\otimes m} X_{i+m+1}.
\end{equation}
The other case for $m$ odd is exactly the same but ends by commuting with $Y_{i+m} Y_{i+m+1}$ to get $X_i Z^{\otimes m} Y_{i+m+1}$. Thus, these strings are contained in the Lie algebra generated by $\cX(2) \cup \cY(2)$.

Furthermore, a simple calculation shows that this set of Pauli strings is closed under commutation with elements of $\cX(2)$ or $\cY(2)$. Thus, this is the full generated sub-Lie algebra.

By counting the elements of the generated algebra, we find the dimension of $\langle \cX(2) \cup \cY(2) \rangle$ as
\begin{equation}
    \sum_{k=2}^{n}2(n-k+1) = n^2-n = n(n-1).
\end{equation}

To prove the final part of the proposition, we see that the single $Z_1$ can be moved around using commutator chains as
\begin{equation}
     Z_k I_{k+1} \stackrel{X_k X_{k+1}}{\longleftrightarrow} Y_k X_{k+1} \stackrel{Y_k Y_{k+1}}{\longleftrightarrow} I_k Z_{k+1},
\end{equation}
and that these individual $Z_i$ then add
\begin{equation}
    X_k X_{k+1} \stackrel{Z_k I_{k+1}}{\longleftrightarrow} Y_k X_{k+1} \stackrel{I_k Z_{k+1}}{\longleftrightarrow} Y_k Y_{k+1}
\end{equation}
so that $\cZ(1) \subset \langle \cX(2) \cup \cY(2) \cup \{Z_1\}\rangle$ and $\cY(2) \subset \langle \cX(2) \cup  \cZ(1) \rangle$, and thus
\begin{align*}
    \langle \cX(2) \cup \cY(2) \cup \{Z_1\} \rangle &= \langle \cX(2) \cup \cY(2) \cup \cZ(1)\rangle = \langle \cX(2) \cup \cZ(1)\rangle. \qedhere 
\end{align*}
\end{proof}

Our third example works with all three of $\cX(2), \cY(2)$ and $\cZ(2)$. Its behaviour depends on whether $n$ is even or odd.
\begin{prop}\label{prop:X2_Y2_Z2}
    The Lie algebra generated by $\cX(2) \cup \cY(2) \cup \cZ(2)$ is given by
    \begin{equation}
        \langle \cX(2) \cup \cY(2) \cup \cZ(2) \rangle
        = Q_n \backslash \{I, X^{\otimes n}, Y^{\otimes n}, Z^{\otimes n} \} = \begin{cases}
            \mf{su}(2^{n-1}) &\text{ if $n$ odd} \\
            \oplus_{i=1}^{4}\mf{su}(2^{n-2}) &\text{ if $n$ even}
        \end{cases},
    \end{equation}
    where $Q_n$ consists of all $n$-qubit Pauli strings $P$ where the $k_x(P) \equiv k_y(P) \equiv k_z(P) \mod{2}$. Note that this implies that $Q_n$ is permutation-invariant. The resulting Lie algebra has dimension $4^{n-1}-1$ for $n$ odd and dimension $4^{n-1}-4$ for $n$ even.
\end{prop}
\begin{proof} In this proof we understand Pauli strings as words in the alphabet $\{X,Y,Z,1\}$. In this sense we first show permutation invariance. Indeed, we obtain all transpositions through the chain of commutators
    \begin{equation}
        ZI \stackrel{XX}{\longleftrightarrow} YX \stackrel{YY}{\longleftrightarrow} IZ \stackrel{XX}{\longleftrightarrow} XY.
    \end{equation}
    This shows that we can permute $XY$ into $YX$ and $ZI$ into $IZ$. Because of symmetry in the generator set, we can make a similar chain of commutators for $XZ$ with $YI$ and $YZ$ with $XI$. Thus, we can apply any transposition, and thus any permutation.
    To prove the parity constraint, we see that the commutator of $XX$ with any $\alpha \beta$ is non-zero if either $\alpha \in \{Y, Z\}$ and $\beta \in \{X,I\}$, or the reverse. Taking the commutator of such a string with $XX$ flips the choice for $\alpha$ and $\beta$. This means that $k_x, k_y$, and $k_z$ all increase or decrease by exactly 1. Taking a commutator with $YY$ or $ZZ$ changes the parity in the same way. Since all the generators are themselves Pauli strings with $k_x \equiv k_y \equiv k_z \equiv 0 \mod{2}$, taking a commutator with another generator preserves the $\mathrm{mod}\: 2$ equivalence of $k_x, k_y$, and $k_z$ and can only flip $k_x, k_y,$ and $k_z$ to all be equal to $1 \mod{2}$.

    Furthermore, an easy calculation shows that $X^{\otimes n}, Y^{\otimes n}$, and $Z^{\otimes n}$ all commute with all the generators. This means that \begin{equation}\label{eq:subset2}
        \langle \cX(2) \cup \cY(2) \cup \cZ(2) \rangle \subset Q_n \backslash \{I, X^{\otimes n}, Y^{\otimes n}, Z^{\otimes n} \}.
    \end{equation} 
A simple but tedious calculation shows that for any Pauli string 
     \[ P \in Q_n \backslash \{I, X^{\otimes n}, Y^{\otimes n}, Z^{\otimes n} \} ,\] 
one can make a Pauli string from $\cX(2), \cY(2)$ and $\cZ(2)$ with the same $k_x(P), k_y(P)$ and $k_z(P)$. The permutation invariance of $\langle \cX(2) \cup \cY(2) \cup \cZ(2) \rangle$ then implies that
    \begin{equation}
        \langle \cX(2) \cup \cY(2) \cup \cZ(2) \rangle \supset Q_n \backslash \{I, X^{\otimes n}, Y^{\otimes n}, Z^{\otimes n} \}.
    \end{equation}
    Together with Eq.~(\ref{eq:subset2}), this proves the second equality. Using the classification method of \cite{Aguilar2024FullClassificationLie}, one directly finds
    \begin{align*}
        \langle \cX(2) \cup \cY(2) \cup \cZ(2) \rangle &= \begin{cases}
            \mf{su}(2^{n-1}) &\text{ if $n$ odd} \\
            \oplus_{i=1}^{4}\mf{su}(2^{n-2}) &\text{ if $n$ even}
        \end{cases}.
    \end{align*}
Finally, we again show that the dimension of $Q_n$ is $4^{n-1}$ by explicitly counting the Pauli strings. One can see this by first considering an arbitrary Pauli string up to index $n-2$, of which there are $4^{n-2}$. Then these strings can be extended by $4\cdot 4=16$ combinations for the last 2 indices, of which $4$ will satisfy the parity constraints, while the other 12 will not. Thus, after removing the inaccessible strings, we find the dimension to be
    \begin{align*}
    \operatorname{dim}\langle \cX(2) \cup \cY(2) \cup \cZ(2) \rangle & =
        \begin{cases}
            4^{n-1}-1 \quad \text{if } n \text{ odd} \\
            4^{n-1}-4 \quad \text{if } n \text{ even}.  
        \end{cases}\qedhere
    \end{align*}
\end{proof}

Having investigated these three basic example Lie algebras, we can now determine what elements to add to them to get the full Lie algebra $\mathfrak{su}(2^n)$.

\begin{theorem}\label{thm:XdoubleZsingle_extensions}
    The Lie algebra generated by $\cX(2)$ and $\cZ(1)$ combined with $X_1$ and/or $Z_1 Z_2$ is given by
    \begin{equation*}
    \begin{aligned}
        i) && \langle \cX(2) \cup \cZ(1) \cup \{X_1\} \rangle =& \{Z_j ~:~ j \in [n]\} \cup \{X_1, Y_1\} \\ 
        &&& \cup \{\alpha_i Z^{\otimes m} \beta_{i+m+2} ~:~ \alpha, \beta \in \{X,Y\},~ m \in [n-2] \} \\
        &&& \cup \{Z^{\otimes m} \beta_{m+1} ~:~ \beta \in \{X,Y\}, m \in [n-1] \}  \\
        &&& = \mf{so}(2n+1)\\
        ii) && \langle \cX(2) \cup \cZ(1) \cup \{Z_1 Z_2\} \rangle =& \langle \cX(2) \cup \cZ(2) \cup \cZ(1) \rangle  \\
        && =& S_{z_n} \backslash \{(Z_1 Z_2 \cdots Z_n), I\} = \oplus_{i=1}^2 \mf{su}(2^{n-1}), \\
        iii) && \langle \cX(2) \cup \cZ(1) \cup \{X_1, Z_1 Z_2\} \rangle =& \PP_n \backslash \{I \} = \mf{su}(2^n)
    \end{aligned}
    \end{equation*}
    with $S_{z_n}$ all Pauli strings with an even number of Pauli $X$ and Pauli $Y$ operators, and $\PP_n$ all possible Pauli strings. The dimensions of the generated Lie algebras are respectively $2n^2+n$, $2( 4^{n-1}-1)$, and $4^n-1$ for $n\geq 3$.
\end{theorem}
\begin{proof}
    The proof for the corresponding abstract Lie algebra follows from \cite{Aguilar2024FullClassificationLie} for all three cases. Here we show the calculations to find the concrete representation corresponding to the given generators.
    
    For assertion $i)$ let us consider the PauLie map which keeps every generator the same, except for $X_1$ being sent to $X_0X_1$. With this generator we are again in $\cX(2)\cup \cZ(1)$, but now acting on one more total qubit. We recall that this Lie algebra is exactly given by the original generators and alternating products. Forgetting the $X_0$ registers is the inverse PauLie map which gives exactly the elements listed above. 
    
    For $ii)$, the extension with $Z_1 Z_2$, we first show by induction that, together with elements of $\cX(2)$ and $\cZ(1)$, we can use $Z_1 Z_2$ to make any element $Z_k Z_{k+1}$. We can interchange any $Z$ with a neighbouring $I$ using the commutator chain
    \begin{equation}
        ZI \stackrel{XX}{\longleftrightarrow} YX \stackrel{YY}{\longleftrightarrow} IZ.
    \end{equation}
    By applying the appropriate transpositions to $Z_1 Z_2$, we can make any $Z_{k} Z_{k+1}$. Thus, we find that 
    \begin{equation}
         \langle \cX(2) \cup \cZ(1) \cup \{Z_1 Z_2\} \rangle = \langle \cX(2) \cup \cZ(2) \cup \cZ(1) \rangle \supset \langle \cX(2) \cup \cY(2) \cup \cZ(2) \rangle.
    \end{equation}
    By using Proposition \ref{prop:X2_Y2_Z2}, we find that the resulting Lie algebra is permutation-invariant. If we track the number of Pauli $X, Y$, and $Z$ operators with $k_x, k_y$, and $k_z$ respectively, we see that the presence of individual $Z$ operators means that $k_x \equiv k_y \mod{2}$ still holds, but that $k_z$ has no such constraints.

    To summarize, for arbitrary $\alpha, \beta \in \{X,Y\}$ and $\alpha', \beta' \stackrel{Z}{\longleftrightarrow} \alpha, \beta$, the following operations on Pauli strings are allowed:
    \begin{equation}
    \begin{aligned}
        I \alpha &\stackrel{ZZ}{\longleftrightarrow} Z \alpha', &  \alpha I & \stackrel{ZZ}{\longleftrightarrow} \alpha' Z \\
        \alpha \beta &\stackrel{\beta \beta}{\longleftrightarrow} ZI , & \alpha \beta &\stackrel{\alpha \alpha}{\longleftrightarrow} IZ \\
        \alpha &\stackrel{Z}{\leftrightarrow} \alpha', \\    
    \end{aligned}
    \end{equation}
    where the first four are the same as in Proposition \ref{prop:X2_Y2_Z2}, and the fifth one is the additional operation that breaks the parity constraint on $k_z$.

    Finally, by chaining together appropriate elements $XX$, $YY$, and $ZZ$, we can make any Pauli string with $k_x+k_y = 2l$ Pauli-$X$ or Pauli-$Y$ matrices, and identity or Pauli-$Z$ matrices on the remaining $n-2l$ components. The only element we cannot make is $Z^{\otimes n}$, as this operator commutes with all the generators.

    The dimension of the generated algebra can be found by summing over an index $k$, which counts the number of $X$ and $Y$ components in the string. For a given $l$, there are $2^{2l}$ options for the $X$ and $Y$, $2^{n-2l}$ options of the remaining $I$ or $Z$, and $\binom{n}{2l}$ options to arrange the $X$ and $Y$ in between the $I$ and $Z$. This gives the total number of options as
    \begin{equation}
        \sum_{l=0}^{\lfloor n/2 \rfloor } 2^{2l} 2^{n-2l} \binom{n}{2l} = 2^n 2^{n-1} = \frac12 4^n,
    \end{equation}
    where we have to subtract 2 for the 2 Pauli strings $I^{\otimes n}$ and $Z^{\otimes n}$ that we can not make to find the dimension of the generated subspace.

    For the last extension with both $X_1$ and $Z_1 Z_2$, the addition of the $X_1$ onto the extension of $Z_1 Z_2$ means that one can now also drop the requirement $k_x \equiv k_y \mod{2}$, and, furthermore, one can also make $Z^{\otimes n}$.
\end{proof}

We refer to section 4 for the implementation of universal gates once we have obtained all Pauli strings.

\begin{rem}\label{rem:Eisert_classification}
    As mentioned at the beginning of this section and in the proofs, these Lie algebras were already classified in \cite{Aguilar2024FullClassificationLie}. In  its algorithmic approach for finding the Lie algebra, one changes the generators by taking commutators while preserving the generated algebra until one ends up in one of a set of known anti-commutation graphs. For the generator set $\cX(2) \cup \cZ(1) \cup \{X_1, Z_1Z_2\}$, this approach finds the equivalent generator set
    \begin{equation}
        \{X_1, Z_1, X_1X_2, Z_2, X_2\} \cup \{X_1 X_k, Z_k ~|~ \text{for } k \in \{3,4,\cdots, n\}\}.
    \end{equation}
    This set matches the commutation relations of the type $B3$ generator sets and proves that the resulting Lie algebra is isomorphic to $\mathfrak{su}(2^n)$, which coincides with our direct calculations. Our Section \ref{sec:math_prelim} results add the information that these maps preserve the distance in the Lie group.
\end{rem}

\subsection{The Ising Hamiltonian as a magic Hamiltonian}\label{sec:Ising_Lie}
Now that we have found the Lie algebras generated by Pauli strings, we can find resource sets $S$ and parameters $\gamma$ and $\lambda$ for which the Ising Hamiltonian
\begin{equation}\label{eq:Ising_Ham}
    J_{\gamma, \lambda} = \sum_j (1+\gamma) X_j X_{j+1} + (1-\gamma)Y_j Y_{j+1} + \lambda \sum_j Z_j
\end{equation}
is magic. Note that we will now have linear combinations of Pauli strings in our generator sets, and thus the symmetry of commutators and the classification of Lie algebras as set up in \cite{Aguilar2024FullClassificationLie} will no longer hold. 

For our first result, we find a resource set $S$ based on $\cZ(1)$ as a corollary of Theorem~\ref{thm:XdoubleZsingle_extensions}.

\begin{cor}\label{Ising}
    Let $J_{\gamma, \lambda}$ as given in Eq.~\eqref{eq:Ising_Ham}. Then for all values of $\gamma$ and $\lambda$
    \begin{equation}
        \langle \{J_{\gamma, \lambda}, X_1, Z_1 Z_2 \}, \cZ(1) \rangle = \langle \{X_1, Z_1 Z_2 \}, \cZ(1), \cX(2) \rangle = \mathfrak{su}(2^n).
    \end{equation}
    Thus, $J_{\gamma, \lambda}$ is magic with respect to $\{X_1, Z_1Z_2\}\cup \cZ(1)$
\end{cor}
\begin{proof}
    For the first equality, one can easily see that 
    \begin{equation}
        \cY(2) \subset \langle \{X_1, Z_1 Z_2 \}, \cZ(1), \cX(2) \rangle,
    \end{equation}
    and thus that $J_{\gamma, \lambda} \in \langle \{X_1, Z_1 Z_2 \}, \cZ(1), \cX(2) \rangle$, which means that
    \begin{equation}\label{eq:subset1}
        \langle \{J_{\gamma, \lambda}, X_1, Z_1 Z_2 \}, \cZ(1) \rangle \subset \langle \{X_1, Z_1 Z_2 \}, \cZ(1), \cX(2) \rangle.
    \end{equation}

    To prove the reverse, we show that any element in $\cX(2)$ is in $\langle \{J_{\gamma, \lambda}, X_1, Z_1 Z_2 \}, \cZ(1) \rangle$ by induction. We can extract $X_1 X_2$ by commuting $J_{\gamma, \lambda}$ with $Z_1$ and $X_1$ as
    \begin{equation}
        J_{\gamma, \lambda} \stackrel{Z_1}{\longrightarrow} \big( (1+\gamma) Y_1 X_{2} - (1-\gamma) X_1 Y_{2}\big) \stackrel{X_1}{\longrightarrow} Z_1 X_{2} \stackrel{X_1}{\longleftrightarrow} Y_1X_2 \stackrel{Z_1}{\longleftrightarrow} X_1 X_2.
    \end{equation}
    
    Now, assuming we can generate $X_k X_{k+1}$, we show that we also find $X_{k+1} X_{k+2}$ in the generated Lie algebra. First take the commutator of $X_k X_{k+1}$ with $Z_{k+1}$ to find $X_k Y_{k+1}$. By taking the commutator of $J_{\gamma, \lambda}$ with $Z_{k+1}$ and $Z_{k+2}$, we get the element
    \begin{equation}
        J_{\gamma, \lambda} \stackrel{Z_{k+1}}{\longrightarrow} \cdots \stackrel{Z_{k+2}}{\longrightarrow}  (1+\gamma) Y_{k+1} Y_{k+2} + (1-\gamma) X_{k+1} X_{k+2}
    \end{equation}
    Now we take the commutator twice with $X_k Y_{k+1}$ to find
    \begin{equation}
        (1+\gamma) Y_{k+1} Y_{k+2} + (1-\gamma) X_{k+1} X_{k+2} \stackrel{X_k Y_{k+1}}{\longrightarrow} X_k Z_{k+1} X_{k+2} \stackrel{X_k Y_{k+1}}{\longrightarrow} X_{k+1} X_{k+2}.
    \end{equation}
    Thus, $\cX(2) \subset \langle \{J_{\gamma, \lambda}, X_1, Z_1 Z_2 \}, \cZ(1) \rangle$, which implies 
    \begin{equation}
        \langle \{J_{\gamma, \lambda}, X_1, Z_1 Z_2 \}, \cZ(1) \rangle \supset \langle \{X_1, Z_1 Z_2 \}, \cZ(1), \cX(2) \rangle.
    \end{equation}
    Together with the result from Eq.~\eqref{eq:subset1}, this proves the first equality of the theorem. The second equality follows directly from Theorem \ref{thm:XdoubleZsingle_extensions}.
\end{proof}

\begin{rem}
    Removing $X_1$ from the generators gives a more complicated effect than seen in Theorem 
 \ref{thm:XdoubleZsingle_extensions}. Specifically, without $X_1$, the symmetry in $J_{0,\lambda}$ for $\gamma = 0$ prevents a separation of elements $(1+\gamma) X_k X_{k+1} + (1-\gamma) Y_k Y_{k+1}$ into the separate $X_k X_{k+1}$ and $Y_k Y_{k+1}$ elements. See also Remark \ref{rem:J_Z2_symmetry}.
\end{rem}

Proposition \ref{prop:double-double} states that $\langle \cX(2) \cup \cZ(1)\rangle = \langle\cX(2) \cup \cY(2)\cup \{Z_1\}\rangle$. That same insight immediately finds us a different $S$ with respect to which $J_{\gamma, \lambda}$ is magic.
\begin{cor}
     $J_{\gamma, \lambda}$ is magic w.r.t. $\{X_1, Z_1 Z_2, Z_1\}\cup \cY(2)$.
\end{cor}
\begin{proof}
    To prove that $J_{\gamma, \lambda}$ is magic w.r.t. $\{X_1, Z_1 Z_2, Z_1\}\cup \cY(2)$, we have to show that $\cX(2)$ is in the generated Lie algebra. We can extract $X_k X_{k+1}$ by commuting with appropriate $Y_k Y_{k+1}$ and/or $Z_1$ and taking linear combinations. Thus, we are back in the scenario of $\langle \cX(2) \cup \cY(2) \cup \{X_1, Z_1 Z_2, Z_1\}\rangle$, which generates $\mathfrak{su}(2^n)$ by Proposition \ref{prop:double-double} and Theorem \ref{thm:XdoubleZsingle_extensions}.
\end{proof}
    
We next investigate extensions of the generators $\cX(2)$ and $\cZ(2)$, attempting to replace the $\cX(2)$ by the Ising Hamiltonian $J_{\gamma, \lambda}$. The interesting case is when we do not add $Z_1$ and try to find different and smaller sets of local operators for which $J_{\gamma, \lambda}$ is magic w.r.t.\ symmetry protection. The symmetry-protected Lie group is everything that commutes with $X^{\otimes n},Y^{\otimes n}$.

\begin{prop}\label{prop:J_doublyZ}
    The Lie algebra generated by $J_{\gamma, 0}$ with $\gamma \notin \{-1,0,1\}$ and $\cZ(2)$ for $n>4$ is given by
    \begin{equation}
    \begin{aligned}
        \langle J_{\gamma, 0} \cup \cZ(2) \rangle &= \langle \cX(2) \cup \cY(2) \cup \cZ(2) \rangle \\
        &= Q_n \backslash \{I, X^{\otimes n}, Y^{\otimes n}, Z^{\otimes n} \} = \begin{cases}
            \mf{su}(2^{n-1}) &\text{ if $n$ odd} \\
            \oplus_{i=1}^{4}\mf{su}(2^{n-2}) &\text{ if $n$ even}
        \end{cases},
    \end{aligned}
    \end{equation}
    where $Q_n$ consists of all Pauli strings of length $n$ where the $k_x \equiv k_y \equiv k_z \mod{2}$. Note that for odd $n$, the parity constraints already ensure that the elements $\{X^{\otimes n}, Y^{\otimes n}, Z^{\otimes n} \}$ are not in $Q_n$. The resulting Lie algebra has dimension $4^{n-1}-1$ for odd $n$ and $4^{n-1}-4$ for even $n>2$.
\end{prop}
\begin{proof}
    We only have to prove the first equality, as the rest follows from Proposition \ref{prop:X2_Y2_Z2}. For $J_{\gamma,0}$, we immediately find that 
    \begin{equation}
        \langle J_{\gamma,0} \cup \cZ(2) \rangle \subset \langle \cX(2) \cup \cY(2) \cup \cZ(2) \rangle.
    \end{equation}
    For arbitrary $\gamma \neq 0$, we can extract the individual $a_k \coloneqq (1+\gamma) X_k X_{k+1} + (1-\gamma) Y_k Y_{k+1}$ for $k \in \{2, \cdots, n-2\}$ by commuting with appropriate $Z_k Z_{k+1}$. We can extract the first $a_1$ and last $a_{n-1}$ by further commutations and linear combinations, which works for $n \neq 4$. Furthermore, with appropriate commutators, we find that
    \begin{equation}
        a_k \stackrel{a_{k+1}}{\longrightarrow} X_k Z_{k+1} Y_{k+2} - Y_k Z_{k+1} X_{k+2} \stackrel{a_{k+1}}{\longrightarrow} (1-\gamma) X_k X_{k+1} + (1+\gamma) Y_k Y_{k+1} 
    \end{equation}
    which means that $(1-\gamma) X_k X_{k+1} +(1+\gamma) Y_k Y_{k+1}$ is also in the generated Lie algebra. If $\gamma \notin \{-1,1\}$, we find that the coefficients are flipped compared to $a_k$, so we can now extract the individual $X_k X_{k+1}$ and $Y_k Y_{k+1}$ by taking linear combinations. Thus
    \begin{equation}
        \langle J_{\gamma,0} \cup \cZ(2) \rangle \supset \langle \cX(2) \cup \cY(2) \cup \cZ(2) \rangle \supset \langle J_{\gamma, 0} \cup \cZ(2) \rangle,
    \end{equation}
    which implies that the two generated Lie algebras are the same.
\end{proof}

In the previous proposition, the growth in terms of dimension of the Lie algebra from adding $J_{\gamma,0}$ was especially big. The elements of $\cZ(2)$ in $S$ all commute, so the Lie algebra originally had dimension $n-1$, while the dimension of $\langle \{J_{\gamma,0} \} \cup \cZ(2) \rangle$ is of the form $4^{n-1}+O(1)$.

\begin{rem}\label{rem:J_Z2_symmetry}
    The cases $\gamma \in \{-1,0,1\}$ are special. For $\gamma = 1$, we have
    \begin{equation}
        J_{1,0} = \sum_k X_k X_{k+1}
    \end{equation}
    and we do not get the $Y_k Y_{k+1}$ elements. Thus, the problem reduces to $\langle \cX(2) \cup \cZ(2) \rangle$, which is polynomial in size as shown in Proposition  \ref{prop:double-double}.

    For $\gamma = 0$, we get the symmetric version. In this case, we can still extract the individual $a_k$ terms, but no longer the individual $X_k X_{k+1}$ and $Y_k Y_{k+1}$ separately. The resulting Lie algebra is exponential in size, as one can make chains by commuting $a_k$ with $a_k+1$ or $Z_{k+1}Z_{k+2}$.
\end{rem}

To generate the full Lie algebra $\mathfrak{su}(2^n)$, we must further break the parity constraints. This can be done by adding full control over a single qubit, using for example $X_1$ and $Y_1$, or setting $\lambda \neq 0$.

\begin{cor}\label{cor:Ising-Zdouble} 
The Lie algebra generated by $J_{\gamma, 0}$ with $\gamma \notin \{-1,0,1\}$ and $\cZ(2) \cup \{X_1, Y_1\}$ for $n>4$ is given by
    \begin{equation}
        \langle \{J_{\gamma,0}\} \cup \cZ(2) \cup \{X_1, Y_1 \} \rangle = \mathfrak{su}(2^n).
    \end{equation}

The Lie algebra generated by $J_{\gamma, \lambda}$ with $\gamma \neq 0$ and $\lambda \neq 0$ and $\cZ(2)$ for $n>4$ is given by
    \begin{equation}
    \begin{aligned}
        \langle J_{\gamma, \lambda} \cup \cZ(2) \rangle &= \langle \cX(2) \cup \cY(2) \cup \cZ(2) \cup \cZ(1) \rangle =\langle \cX(2) \cup \cZ(1) \cup \{Z_1 Z_2\} \rangle = S_{Z_n}\backslash \{Z^{\otimes n}, I\},
    \end{aligned}
    \end{equation}
    where $S_{Z_n}$ is all Pauli strings with an even number of Pauli $X$ and Pauli $Y$ operators. 

The Lie algebra generated by $J_{\gamma, \lambda}$ with $\gamma \neq 0$ and $\lambda \neq 0$ and $\cZ(2) \cup \{X_1\}$ for $n>4$ is given by
    \begin{equation}
        \langle J_{\gamma, \lambda} \cup \cZ(2) \cup \{X_1\}\rangle =\mathfrak{su}(2^n).
    \end{equation}
\end{cor}
\begin{proof}
    The first claim follows by seeing that adding $X_1$ and $Y_1$ gives all elements in $\cX(1), \cY(1)$ and $\cZ(1)$, after which Theorem \ref{thm:XdoubleZsingle_extensions} shows we generate the full Lie algebra.

    For the second claim with $\lambda \neq 0$, we can prove the first equality similarly to the proof of Proposition~\ref{prop:J_doublyZ}. We first extract the individual $X_k X_{k+1}$ and $Y_k Y_{k+1}$ terms by taking commutators and linear combinations with $Z_k Z_{k+1}$. We can then use commutators with these $X_k$ to extract the individual $Z_k$ terms, after which the first equality follows. The second equality $ \langle \cX(2) \cup \cY(2) \cup \cZ(2) \cup \cZ(1) \rangle =\langle \cX(2) \cup \cZ(1) \cup \{Z_1 Z_2\} \rangle$ follows easily from seeing that $Y_k Y_{k+1}$ and $Z_k Z_{k+1}$ are in $\langle \cX(2) \cup \cZ(1) \cup \{Z_1 Z_2\} \rangle$. The last equality follows from Theorem \ref{thm:XdoubleZsingle_extensions}.

    The last claim of extension with $X_1$ follows from the second claim and Theorem \ref{thm:XdoubleZsingle_extensions}.
\end{proof}

We summarize all our results up to this point in the following theorem.
\begin{theorem}\label{themagic}
In the following table, we give various local control sets $S$, their dimension, and a candidate magic Hamiltonian.
 \[ \begin{array}{c|c|c|c|c}
 \text{No.}& \text{local set } S & \text{Dim}(S) & \text{Hamiltonian} & \text{magic} \\  \hline
 1)& \mathcal{X}(2)\cup \mathcal{Z}(1)& 2n^2-n& X_1+Z_{1}Z_2 & \text{yes}  \\ 
 2)& \{X_1\}\cup \mathcal{X}(2)\cup \mathcal{Z}(1) & 2n^2+n & Z_1Z_2 & \text{yes} \\ 
 3)& \mathcal{Z}(1)  & n & \sum_j X_jX_{j+1} &  \text{no}\\
  4)&\{X_1,Z_1Z_2\}\cup \mathcal{Z}(1) & n+5 & \sum_j X_jX_{j+1} &  \text{yes}\\
  5)&\{X_1,Z_1Z_2\}\cup \mathcal{Z}(1) & n+5 & J_{\gamma,\lambda} &  \text{yes}\\
  6)& \{X_1,X_1X_2X_3X_4,Z_1,Z_1Z_2Z_3Z_4\} \cup \mathcal{Z}(2) & n+32 & \sum_j X_jX_{j+1} & \text{no} \\ 
  7)& \{X_1,X_1X_2X_3X_4,Z_1,Z_1Z_2Z_3Z_4\} \cup \mathcal{Z}(2) & n+32 & \sum_j X_jX_{j+1} + \sum_j Z_j& \text{yes} \\ 
  8)&\{X_1\}\cup \mathcal{Z}(2)  & n+1 & J_{\gamma,\lambda} & \text{yes}^* \\ 
 9)& \{X_1,Y_1\} \cup \mathcal{Z}(2) & n+3 & J_{\gamma,0} & \text{yes} 
  \end{array} \] 
Here yes$^*$ is under the condition $\la\neq 0$, $\gamma\notin\{-1,0,1\}$. \label{thm:controlsets}
\end{theorem} 
Note that for case $8)$ and $\gamma =1$, we find $J_{\gamma, \lambda} = 2\sum_j X_j X_{j+1}$ and the generated algebra is
\[
\langle \{X_1, J_{\gamma, \lambda}\}\cup \cZ(2)\rangle = \langle \cX(2) \cup \cZ(2) \cup \{X_1\} \rangle = \mf{so}(n)\oplus\mf{so}(n+1),
\]
with dimension $\frac12n(n+1)$, and for $\gamma=-1$, we find $J_{\gamma, \lambda} = 2\sum_j Y_j Y_{j+1}$ and the generated algebra is
\[
\langle \{X_1, J_{\gamma, \lambda}\}\cup \cZ(2)\rangle = \langle \cX(1) \cup \cZ(2) \rangle = \mf{so}(2n),
\]
with dimension $2n^2-n$. Thus, the change from $\gamma \in \{\pm 1\}$ to $\gamma \notin \{\pm 1\}$ changes the scaling of the dimension of the generated Lie algebra from polynomial to exponential. This phase transition matches to the operator algebra phase transition for the Ising model as studied by Araki and Matsui \cite{AM};  we discuss this point further in Appendix~\ref{app:operator_algebra}.

\begin{proof}[Proof of theorem~\ref{themagic}] We deduce 2) from the last statement of Theorem \ref{thm:XdoubleZsingle_extensions}. 

1) follows because  double commutation of $X_1+ Z_1 Z_2$ with $Z_1$ produces $Z_1Z_2$, and hence also $X_1$ by linear combination, and we are back in the situation of 2).  

For 3) we  stay in the Lie algebra generated by $\mathcal{Z}(1)\cup \mathcal{X}(2)$ which has polynomial dimension. 

4)  and 5) follow directly from Corollary~\ref{Ising}.

The case 6) and 7) starts with the algebra $\langle \mathcal{X}(2)\cup \mathcal{Z}(2) \rangle = \langle \tilde{S}_1\cup \tilde{S}_2 \rangle $ similar to Eq.~\eqref{2.20}. Let us assume $n$ is even, then 
\[  
    S_1\lel \{X_1X_2,Z_2Z_3,X_3X_4,\cdots  ,X_{n-1}X_n\}\pl \text{and}\pl \tilde{S}_1=\{X_1, Z_1 Z_2, X_2,\cdots, Z_{n/2-1}Z_{n/2}, X_{n/2}\} 
\] 
generate the same Lie algebra, Lie group and $C^*$-algebra. One can generate $\oplus_{i=1}^2\mf{su}(2^{n/2})$ by adding $Z_1$ and $X_1 X_2$ to $\tilde{S}_1$ and $X_{n/2+1}$ and $Z_{n/2+1} Z_{n/2+2}$ to $\tilde{S}_2$. Now we can use Lem.~\ref{lem:PauLie_homomorphism}, which states that if a mapping between Pauli strings is PauLie (i.e.\ it preserves the commutation relations), the generated Lie algebra is the same. Thus, to extend $S_i$, we have to add elements that satisfy the same commutation relations, from which we find that extending $S_1$ with $\{Z_1, X_1X_2X_3X_4\}$ and $S_2$ with $\{X_1, Z_1 Z_2 Z_3 Z_4\}$ will result in
\begin{equation}
    \langle S_1 \cup \{Z_1, X_1 X_2 X_3 X_4\} \cup S_2 \cup \{X_1, Z_1 Z_2 Z_3 Z_4\} \rangle = \oplus_{i=1}^2 \mf{su}(2^{n/2}).
\end{equation}
Now adding in a magnetic term $\sum_j Z_j$ to the Hamiltonian, we find 
\begin{align*}
    \langle S_1 \cup S_2 \cup \{X_1, Z_1, X_1 X_2 X_3 X_4, Z_1 Z_2 Z_3 Z_4\} \cup \cZ(1) \rangle &= \langle \cX(2) \cup \cY(2) \cup \cZ(2) \cup \{X_1, Z_1\} \rangle \\
    &=  \mf{su}(2^n)
\end{align*}
by Corollary \ref{cor:Ising-Zdouble}. The argument for odd $n$ is similar.
For 8) and 9) we refer again to Corollary \ref{cor:Ising-Zdouble}. 
\end{proof} 

\section{Algorithmic implementation} \label{sec:algorithmic_implementation}
In the previous section, we derived the Lie algebras generated by various generator sets. These sets represent the physical controls one might have on a quantum computer; the generated Lie algebra is then accessible control space, and the corresponding Lie group the set of unitaries one can actually apply to evolve the state of the physical quantum computer. To actually implement an arbitrary desired unitary operation using only the supplied physical controls is not a trivial task theoretically (not to mention experimentally), but it can be done using the aid of the existing literature, in particular the geometric work of Chow-Rashevskii.

The first step of implementing an arbitrary unitary $U$ is to pass to the Lie algebra by finding $H$ s.t. $U = \exp(iH)$. In our setting, we will always have a set of physical controls $S$ such that $U \in SU(2^n)$ and $H \in \mf{su}(2^n) = \langle \{S_j\} \rangle$. One finds a basis $\{P_k\}$ of $\langle \{S_j\} \rangle$ s.t. each basis element $P_k$ can be written as a chain of commutators of elements $S$:
\begin{equation}
    P_k = [S_{j_1}, [S_{j_2}, \cdots [S_{j_{n-1}}, S_{j_n}]]].
\end{equation}
In this basis, one can rewrite the control $H$ as a weighted sum over those carefully chosen basis elements $P_k$:
\begin{equation}
    H = \sum_k h_k P_k.
\end{equation} 
One can then implement the basis expansion of $\exp(H)$ using Trotterization and the chain of commutators using the constructive proof of the Chow-Rashevskii theory given in \cite{Giannotti2024ProvingRashevskii}. The Trotterization is implemented as
\begin{equation}
    \exp(zH) = \exp\big(z\sum_k h_k P_k\big) = \lim_{n\to \infty} \Big(\prod_k \exp(z h_k P_k/n) \Big)^n.
\end{equation}
The Chow-Rashevskii implementation can be implemented using an iterative scheme that uses
\begin{equation}
    \exp(itA)\exp(itB) = \exp\big( it(A+B) - \frac{t^2}{2}[A,B] + \cO(t^3)\big)
\end{equation}
to make a commutator by alternating $A$ and $B$ as
\begin{equation}
    \exp(-i\sqrt{t}A)\exp(-i\sqrt{t}B)\exp(i\sqrt{t}A)\exp(i\sqrt{t}B) = \exp\big(t[A,B] + \cO(t^{3/2})\big).
\end{equation}
Using a clever basis of operators, one can extend this to nested commutators using smaller roots of $t$. This is a result of the ball-box theorem in mathematics, see for example \cite{Nagel1985}. Applying the gate $\exp(-i\sqrt{t}A)$ might not be physically possible, but generally $\exp(i(2\pi - \sqrt{t})A)$ will implement the same gate in the periodic case. 

To be able to apply this iterative scheme, one does need an explicit chain of nested commutators for each of the basis elements. If the generators are a subset of all Pauli strings, then Lemma \ref{lem:paulis_closed} shows that one can use the Pauli strings as their basis $P_k$ of $\mf{su}(2^n)$, as then such a chain of nested commutators will always exist. If the generators are not Pauli strings, the approach above will still work, but one might need a different basis when the Pauli strings are not be expressible as a nested chain of commutators. A simple example of this is the Lie algebra generated by $X+Y$ and $X-Y$, which is $\mf{su}(2)$. This Lie algebra contains the Pauli string $X$, but this element can not be written as a nested commutator of $X+Y$ and $X-Y$ without taking linear combinations.

We first consider the settings of Theorem \ref{thm:XdoubleZsingle_extensions}, where we are able to give an explicit algorithm to write any Pauli string as a chain of nested commutators using the elements of the generator set.
\begin{theorem}\label{thm:commutators_algorithmically}
    In the settings of Theorem \ref{thm:XdoubleZsingle_extensions}, any Pauli string can be written as a chain of commutators using the generators
    \begin{equation}
        \langle \cX(2) \cup \cZ(1) \cup \{X_1\} \cup \{Z_1 Z_2\} \rangle = P_n \backslash \{I \},
    \end{equation}
    where the chain has length at most $4n^3$ and can be found algorithmically.
\end{theorem}

\begin{proof}
We provide an explicit algorithm, which we build using three local subroutines. Each subroutine starts with a specific type of Pauli string and applies a sequence of commutators with elements of the generator set to change the Pauli string. We will denote $W_{\leq k}$ for an arbitrary Pauli string on indices $1,2,\ldots, k$ and $W_{\geq k}$ for an arbitrary Pauli string on indices $k, k+1, \ldots, n$. The identity operator $I$ acts on any indices not mentioned.

The three routines are
\begin{equation*}
\begin{aligned}
    1)&& \quad W_{\leq k-1} Z_k I_{k+1} W_{\geq k+2} &\to W_{\leq k-1} I_k Z_{k+1} W_{\geq k+2}, \\
    2)&& \quad Z_1 I_2 W_{\geq 3} & \to Z_1 Z_2 W_{\geq 3}, \\
    3) &&\quad Z_1 W_{\geq k+1} & \to Z_1 X_k W_{\geq k+1}.
\end{aligned}
\end{equation*}
The first subroutine interchanges a $Z$ operator in a Pauli chain with a neighbouring $I$ operator and the second subroutine adds an extra $Z$ operator at the beginning of the Pauli string. These two subroutines together allow us to create any Pauli string consisting of only $I$ and $Z$. The third subroutine makes a $X$ or $Y$ at position $k$ without touching any of the elements past index $k$.

These three subroutines together allow us to make any Pauli string by starting with $Z_1$ and then creating any Pauli string iteratively from highest to lowest index. If the Pauli string has leading $I$'s, we can get rid of the final $Z_1$ as well by either simply shifting it over with the first subroutine if the leading non-identity $\alpha_k = Z$ or using a slightly modified version of the third subroutine that still starts with $Z_1$ but makes $X_k$ instead of $Z_1 X_k$. If there are no leading $I$'s, the string on the first two indices $\alpha_1 \alpha_2$ can be fixed using commutators with $X_1, Z_1, Z_2, X_1 X_2,$ and $Z_1 Z_2$, as these operators fully control the first two qubits.

The first subroutine starts with $W_{\leq k-1} Z_k I_{k+1} W_{\geq k+2}$. Taking the commutators
\begin{equation}
    Z_k I_{k+1} \stackrel{X_k X_{k+1}}{\longleftrightarrow} Y_k X_{k+1} \stackrel{Z_k}{\longleftrightarrow} X_k X_{k+1} \stackrel{Z_{k+1}}{\longleftrightarrow}  X_k Y_{k+1} \stackrel{X_k X_{k+1}}{\longleftrightarrow} I_k Z_{k+1},
\end{equation}
the subroutine shifts one $Z$ one index at the cost of 4 commutations, which is the same as adding 4 to the length function $l$ defined in Definition \ref{def:length_function}.

The second subroutine starts with $Z_1 I_2 W_{\geq 3}$ and uses a sequence of commutations
\begin{equation}
    Z_1 I_2 W_{\geq 3} \stackrel{X_1}{\longleftrightarrow} Y_1 I_2 W_{\ge 3} \stackrel{Z_1 Z_2}{\longleftrightarrow} X_1 Z_2 W_{\ge 3} \stackrel{Z_1}{\longleftrightarrow}Y_1 Z_2 W_{\ge 3} \stackrel{X_1}{\longleftrightarrow} Z_1 Z_2 W_{\geq 3}.
\end{equation}
to increase the number of $Z$ matrices in the Pauli chain by 1 at the cost of 4 commutations. Together with the first subroutine, this allows us to make any string consisting of only $I$ and $Z$. The cost of creating $Z_k$ instead of $I_k$ is 4 commutations to make an extra $Z$ at index 2, and $4(k-2)$ to shift the index over from $Z_2$ to $Z_k$, for a total of $4(k-1)$ commutations.

The third and final subroutine creates an element $Z_1 Y_k W_{\geq k+1}$ or $Z_1 X_k W_{\geq k+1}$ from $Z_1 W_{\geq k+1}$. 
The subroutine starts with $Z_1 I_2 Z_3 Z_4 \cdots Z_k W_{\geq k+1}$, which can be made from $Z_1 W_{\geq k+1}$ using the previous two subroutines for $\sum_{j=3}^k 4(k-1) = 2k(k-1)-4$ commutations. Then the subroutine adds $X$ to the second index via the commutators
\begin{equation}
    Z_1 I_2 Z_3 Z_4 \cdots Z_k W_{\geq k+1} \stackrel{X_1 X_2}{\longleftrightarrow} Y_1 X_2 Z_3 Z_4 \cdots Z_k W_{\geq k+1} \stackrel{X_1}{\longleftrightarrow} Z_1 X_2 Z_3 Z_4 \cdots Z_k W_{\geq k+1}.
\end{equation}
After that, the subroutine shifts the $X$ to the end of the chain of $Z$'s by eating up the $Z$ matrices. This is implemented by commuting with $X_j X_{j+1}$ and $Z_{j+1}$ for $j = 2, 3 \cdots k-1$ to get
\begin{equation}
\begin{aligned}
        Z_1 X_2 Z_3 \cdots Z_kW_{\geq k+1}  &\stackrel{X_2 X_3}{\longleftrightarrow} Z_1 I_2 Y_3 \cdots Z_kW_{\geq k+1}  \stackrel{Z_3}{\longleftrightarrow} Z_1 I_2 X_3 Z_4 \cdots Z_kW_{\geq k+1} \\
        &\cdots  \\
        &\stackrel{X_{k-1} X_k}{\longleftrightarrow} Z_1 I_2 \cdots I_{k-1} Y_k W_{\geq k+1} \stackrel{Z_k}{\longleftrightarrow} Z_1 I_2 \cdots I_{k-1} X_k W_{\geq k+1}. 
\end{aligned}
\end{equation}
Note that applying the last $Z_k$ instead yields  $Z_1 Y_k W_{\geq k+1}$. In total, the cost of creating $Z_1 X_k W_{\geq k+1}$ from $Z_1 W_{\geq k+1}$ is $2k(k-1) - 4 + 2 + \sum_{j=2}^{k-1} 2 = 2k^2-4$ and the cost of creating $Z_1 Y_k W_{\geq k+1}$ from $Z_1 W_{\geq k+1}$ is $2k^2-5$.

This cost calculation does not change if the leading symbols are identities. In that scenario, we can use a modified version of the third subroutine, where we start with
\begin{equation}
    I_1 X_2 Z_3 Z_4 \cdots Z_k W_{\geq k+1}.
\end{equation}
Applying then the commutators with $X_j X_{j+1}$ and $Z_{j+1}$ for $ k = 2, \cdots, k-1$ results in
\begin{equation}
    I_1 I_2 \cdots I_{k-1} X_k W_{\geq k+1}.
\end{equation}

Note that this approach is not necessarily the most efficient one in how the commutators are applied at the first 2 indices, but gives a bound on the necessary number of commutators. The most expensive Pauli string under this algorithm consists of $X_1 X_2 \cdots X_n$ and costs
\begin{equation}
    2+\sum_{k=2}^n (2k^2-4) = \frac23 n^3 +O(n^2).
\end{equation}
See Appendix \ref{app:alg-computation} for a pseudocode writeup of the described approach.\end{proof}

For this generator set, we thus need at most $O(n^3)$ nested commutators to exactly implement any Pauli string. Comparing this to the natural control methods of typical physical quantum computers, one can see a scaling difference. On a quantum computer, one generally has full control over the individual qubits and some entanglement operation, which corresponds to generators similar to $\cZ(1), \cX(1)$ and $\cZ(2)$. Using commutators of these generators, one can make any Pauli string up to length $k$ in $O(k)$ number of commutators. So what our magic Hamiltonian wins in the reduced number of generators for full expressibility, it loses in total control time. Whether this tradeoff is desirable or not will of course depend on the physical system, ease of evolving with the magic Hamiltonian, etc. We discuss this point further in the last Section.

\begin{corollary}
    One can find an explicit algorithm for Corollary \ref{Ising} (Case (5) in theorem \ref{thm:controlsets}), with resource set $\langle J_{\gamma, \lambda} \cup \{X_1, Z_1 Z_2\} \cup \cZ(1)\rangle$, by adding a fourth subroutine that does the translation from $\langle \{J_{\gamma, \lambda}, X_1, Z_1 Z_2 \}, \cZ(1) \rangle$ to $\langle \{X_1, Z_1 Z_2 \}, \cZ(1), \cX(2) \rangle$.
\end{corollary}
\begin{proof}
    The fourth subroutine can be set up in various ways. In the most straightforward setting, one can follow along with the proof of Corollary \ref{Ising} and use iterated commutators to create $X_k X_{k+1}$ from $J_{\gamma, \lambda}$. However, since this is done iteratively and uses $X_k X_{k+1}$ twice to create $X_{k+1} X_{k+2}$, the commutator cost of a single $X_k X_{k+1}$ becomes exponential in $k$. 
    
    There is a way to keep the number of nested commutators polynomial, depending on $\gamma$. For $\gamma \in \{-1,1\}$, one can create each element of $\cX(2)$ and $\cY(2)$ directly by commuting with the appropriate $Z_i$. For $\gamma \notin \{-1,0,1\}$, one can instead create each $a_k \coloneqq (1+\gamma) X_k X_{k+1} + (1-\gamma) Y_k Y_{k+1}$ and the flipped version $a_k' \coloneqq (1-\gamma) X_k X_{k+1} + (1+\gamma) Y_k Y_{k+1}$ by 
    \begin{equation}
        a_k \stackrel{a_{k+1}}{\longrightarrow} X_k Z_{k+1} Y_{k+2} - Y_k Z_{k+1} X_{k+2} \stackrel{a_{k+1}}{\longrightarrow} (1-\gamma) X_k X_{k+1} + (1+\gamma) Y_k Y_{k+1} \eqqcolon a_k'.
    \end{equation}
    Now $a_k$ and $a_k'$ span the same space as $X_k X_{k+1}$ and $Y_k Y_{k+1}$, so one can set up the same polynomial algorithm but in a different basis build around the $a_k$ and $a_k'$.
\end{proof}

The basis-transformed algorithm based on $a_k$ and $a_k'$ will also work for the setting of Corollary~\ref{cor:Ising-Zdouble} with resource set $\langle J_{\gamma, \lambda} \cup \{X_1\} \cup \cZ(2)\rangle$. One can create $a_k$ directly by commuting with appropriate elements of $\cZ(2)$ and then create $a_k'$ in the same way as above.

\section{Alternating approximation length and covering numbers}\label{sec:approx_length}
In this section we make more precise our previous observation that a magic Hamiltonian has to be used often to approximate an arbitrary unitary. We first define our length measure.
\begin{defi}\label{def:alt_approx_length}
    Given a local Lie algebra $G_0$ and a magic Hamiltonian $H$, we define the alternating approximation length $l_H(\epsilon|G_0)$ as
    \begin{equation}
        l_H(\eps|G_0)\equiv\inf\{m ~:~ \forall u \in U(2^n) ~\exists u_{2j}=e^{it_jH} , u_{2j+1}\in G_0 \mbox{ such that } \|u-u_1u_2\cdots u_{2m+1}\|<\eps  \}.
    \end{equation}
\end{defi}
When the context is clear, we may just write $l_H$ or $l_H(\eps)$. There are other canonical measures $l_H^p(\eps|G_0)=\inf (\sum_j |t_j|^p)^{1/p}$ where the infimum is taken over the same configurations as above. The approximation length corresponds to the $l^0$-norm, measuring the cardinality of the support of the sequence $(t_2,...,t_{2m})$. 

\subsection{Facts about covering numbers for groups}

The topological notion of covering numbers will prove a crucial tool in our results. For a compact subset $G$ in a normed space (or metric space), the covering number is defined as    
  \[ N(G,\eps) \lel \inf\{|S| ~:~ \exists_{S} \forall_{u\in G} \exists_{s\in S}  \text{ s.t. }d_{\|\pl \|}(u,s)<\eps\} \pl .\] 
Every finite set $S$ satisfying  this covering condition is called an  $\eps$-net. 

\begin{theorem}(Szarek \cite{Szarek1998MetricSpaces})\label{Szarek}  Let $d$ be the real dimension of $G$. When $G$ is one of $O(M)$ or $U(M)$, then 
  \[   (c/\eps)^d \kl N(G,\eps) \kl (C/\eps)^d \] 
 holds for universal constants $c,C>0$.  
 \end{theorem}

\begin{rem}\label{Lipschitz} Let us state the main inequalities for the exponential map. The upper estimate  
 \begin{equation}\label{up}
  \|e^{a}-e^{b}\|\kl \|a-b\| (1+e^{\|a\|}e^{\|b\|})
  \end{equation} 
holds in Banach algebras. For the unitary group we may assume $u=e^{ia}$ with $\|a\|_{\infty}\le \pi$ and hence $(1+e^{2\pi})$ is a bound for the  Lipschitz constant. The lower estimate
 \[ \frac{\|e^{ia}-e^{ib}\|}{\|a-b\|} \gl 0.4 \]
was proved in \cite[Lemma 4]{Szarek1998MetricSpaces} assuming that $\|a\|_{\infty} \le \frac{\pi}{4}, \|b\|\le \frac{\pi}{4}$, using the fact that  $\|e^{ia}-1\|
\le \|a\|$ for $\|a\|\le \pi$. This shows that in a fixed neighbourhood of the identity the exponential map is bi-Lipschitz for what might be called `full groups' such as $SO(n)$ and $SU(n)$. For general Lie subgroups $G\subset U(n)$ we only have the lower estimate 
\[ (\frac{c}{\eps})^d \kl N(G_0,\eps)  \pl. \]
Unfortunately, we usually  need an upper bound for the local Lie groups.  
\end{rem} 

We need a few basic facts about the covering number in metric spaces. 

\begin{prop}\label{basic} Let $\eps\le 1$. 
 \begin{enumerate}
  \item[i)] Let $K_1\subset K_2$ be a subset, then $N(K_1,\eps)\kl N(K_2,\frac{\eps}{2})$
  \item[ii)]  Let $g_n$ be the Lie algebra of $n$ commuting Pauli strings, whose Lie group is the torus $\mathbb{T}^n$. Then   \[ N(g_n ,\eps) \kl (\frac{C}{\eps})^{2n} \pl. \]
  \item[iii)] For every Hamiltonian $H\in \mz_{2^n}$ and the one-parameter group $G_H=\{e^{itH}|t\in \rz\}$,
   \[ N(G_H,\eps) \kl (\frac{C}{\eps})^{2^{n+1}} \pl. \]    
    \item[iv)] Let $g=g_1\oplus \cdots \oplus g_m$ be a Lie algebra with commuting sectors $[g_j,g_k]=0$. Then    
    \[ N(G,\eps) \kl \prod_j N(G_j,\frac{\eps}{m}) \pl. \]
 \end{enumerate}
\end{prop}

\begin{proof} i) is well-known \cite{PSV,KSV}. Indeed, for any $\eps/2$ covering $S_2$ of the large set $K_2$ we define $S_1$ as the elements $s\in S_1$ such that for some $x_s\in K_1$ we have $d(s,x_s)<\frac{\eps}{2}$.  Then $N(K_1, \eps)\le |\{x_s|s\in S_1\}|\le |S_2|$. For ii) we note that $\mathbb{T}^n\subset {\rm Ball}_{\ell_{\infty}^n}$ is a subset of a complex Banach space of real dimension $d=2n$. The standard volume estimate \cite{PSV} gives $N({\rm Ball},\eps)\kl (1+\frac{2}{\eps})^{2n}$. Note that for every self-adjoint $H$ we can find a basis of the eigenvectors such that $H$ is diagonal. This implies that $e^{itH}$ is a subset of $\mathbb{T}^{2^n}$, so again the volume estimate proves iii).  For iv) we observe that $G=G_1\times \cdots \times G_m$ is just a product group. For unitaries we have 
  \[ \|u_1\cdots u_m-v_1\cdots v_m\|
  \kl \sum_{j=1}^m \|u_1\cdots u_{j-1}(u_j-v_j)v_{j+1}\cdots v_m\|
  \kl \sum_{j=1}^m \|u_j-v_j\| \pl. \]
Therefore we can choose $\eps/m$-nets  in the subgroups and deduce the assertion.   \end{proof} 

\subsection{Universal sets}\label{sub:univeral_sets}
We now discuss the relation of our Hamiltonian framework to the Solovay-Kitaev theorem. 

\begin{prop}\label{finite} Let $H$ be a magic Hamiltonian with respect to $S_0$. Then $l_H(\eps|S_0)$ is finite. Moreover, a finite set of times for which $H$ is evolved is sufficient.  \end{prop}

\begin{proof} Let us start with a basis $B=\{b_1,...,b_k\}$ of the Lie algebra of $\mathfrak{su}(2^n)$. We may define the norm 
 \[ \|\sum_j \al_j b_j\|_1  \lel \sum_j |\al_j|  \pl. \]
Since we are dealing with a finite-dimensional real Banach space, we find a $\delta$-net $N({\rm ball}_1,\delta)\kl (C/\delta)^{m}$. Of course, this norm is equivalent to the operator norm, i.e. 
  \[ \frac{1}{c_1} \|a\|_{\infty} \kl \|a\|_1 \kl c_2 \|a\|_{\infty} \pl .\]
As usual, we consider the exponential map $\phi(a)=e^{a}$. The set 
\[ K=\{a| \|a\|_{\infty}\le \pi\}\subset \frac{1}{c_1}{\rm Ball}_1\pl\]
is compact and contained in $C_2{\rm Ball}_{\|\pl\|_1}$ and hence admits a $\delta$-net $S$ in the $\|\pl\|_1$-norm of cardinality $(1+\frac{2c_1}{\delta})^{k}$. According to  Remark  \ref{Lipschitz}, we deduce that $\exp(S)$ is a $(1+2^{2\pi})\delta$-net for $G=U(2^n)$. We choose $(1+e^{2\pi})\delta=\frac{\eps}{2}$. Since the set $S$ is finite it suffices to find an alternating approximation for every individual $s=\sum_j \al_j b_j$. First we apply Trotter's formula \cite{Suzuki1976}:
 \[ \|e^s-(\prod_{j=1}^k e^{\al_j b_j/N})^N\| \kl \frac{2}{N} (\sum_{j=1}^k |\al_j| \|b_j\|)^2 \exp(2 \sum_{j=1}^m |\al_j| \|b_j\|) \pl .\]
Let $C'=\max_j\|b_j\|$ and $\al=\frac{C'}{c_1}$. Then we can choose $N\gl \frac{8}{\eps \al}$ and obtain an $\eps/4$-approximation of  the exponential. Since $N$ and $k$ are now fixed, it suffices to find an approximation of $e^{\al_j b_j}$. This is exactly given by the concrete formula for the Chow-Rashevskii theorem \cite{Giannotti2024ProvingRashevskii} for some $l=l(\al_s,b_j)$, depending on the designed error $\frac{\eps}{4Nk}$. Taking the supremum of $kN(s)l(\al_j,b_j)$ provides a total bound in terms of a given bracket-generating set. 
Recall that for a magic Hamiltonian $H$ the set $S_0\cup H$ is bracket-generating. Of course the total number of generators used in this approximation is an overcount of the actual alternating usage. However, in all our approximations only a finite number of terms $e^{it_jH}$ have been used. Collecting these $t_j$ in a set $T$ yields the last assertion. \end{proof}
 
\begin{cor} Let $H$ be a magic Hamiltonian with respect to a generating set $S_0$ s.t. $G_0 = \langle S_0 \rangle$.
Let $g_0^{-1}=g_0\subset G_0$ be a set generating a dense subgroup of $G_0$. Then there exists a finite set $T\subset \rz$ such that
$g=g_0\cup \{e^{itH}|t\in T\}$ is universal.
\end{cor}

\begin{proof} Let $\eps$ be the $\eps$ from the Solovay-Kitaev theorem \cite{KSV}[Theorem 8.5]. We apply Proposition \ref{finite} to $\eps/2$. Then we keep the finite instances of $t_j$ in front of $H$ and replace $e^{it_jh_j}$ for $h_j\in S_0$ by products of elements of $g_0$ allowing an error $\frac{\eps}{2Nkl(s,a_j)}$ at most. Then we have created an $\eps$-net in which all of the elements in the net are finite products of $g_0$ and $e^{it_jH}$. According to the Solovay-Kitaev theorem, we can now approximate every unitary up to precision $\delta$ with products of length $O(\ln^{\al} \frac{1}{\delta})$ for some $\al>1$. Thus $g_0\cup \{e^{it_jH}|t_j \in T\}$ is universal.  
\end{proof}

\begin{rem}\label{indep} The Solovay-Kitaev theorem also tells us that for  $\delta<\eps$  the alternation length only differs by a factor $\log^{\al}(\frac{1}{\delta})$. Since we are aiming for exponential lower bounds we can work with $\delta=2^{-poly(n)}$ without changing the behaviour considerably.     
\end{rem}
\subsection{Lower bounds for alternating length}

In order to provide lower estimates of the alternating length, we will need upper estimates for covering numbers of `small' subgroups in particular spin groups. In our first example with 
 \[ S_0 \lel \{X_1\}\cup \{X_1X_2, \cdots, X_{n-1}X_n\} \cup \{Z_1, \cdots, Z_n\}, \] 
we have to understand the Lie group $\mathfrak{g}_{loc}$, the Lie group generated by the exponential of the Lie algebra $\mathfrak{g}_{loc}$ of local Hamiltonians. We will use the Araki-Jordan-Wigner transform as explained in \cite{API}. We define
 \[ S_j\lel Z_1\cdots Z_{j-1} \]
and
 \[ \gamma_{2j-1} \lel S_jX_j \pl ,\pl \gamma_{2j}\lel S_jY_j \pl .\]   
 The $\gamma_j$'s are self-adjoint and satisfy the Majorana (Clifford)  commutation relation 
  \[ \gamma_j\gamma_k+\gamma_k\gamma_j \lel 2\delta_{jk} \pl.\] 
Therefore the vector space
 \[ E^{(2)} \lel {\rm span} \{\gamma_j\gamma_k |1\le j,k\le 2n\} \] 
is exactly the Lie algebra ${\rm spin}(2n)$. The corresponding Lie group $Spin(2n)$, a double cover of $SO(2n)$, is well-known in geometry and physics. Let us spell out the nature of the cover more precisely. Since the spin group consists of unitaries in $\mz_{2^n}$ we can use the conjugation action. Indeed, since the linear map $\gamma(v)=\sum_j v_j\gamma_j$ is a real linear isometry, the spin products leave the space $\gamma(V)$ invariant. This implies that for every $\xi\in {\rm Spin}$ there exists $o(\xi)\in O(m)$ such that 
\[ \xi \gamma(v)\xi^* \lel \gamma(o(\xi)(v)) \pl .\]
The map $o:{\rm Spin}(m)\to O(m)$ is a surjective group homomorphism with kernel ${\rm ker}(O)=\zz_2$ (see \cite{Gar}). Equivalently ${\rm Spin}(m)=O(m)\rtimes \zz_2$ is given by a semidirect product, sometimes referred to a double cover. As such ${\rm Spin}(m)$ embeds as  unitaries into $\mz_m\oplus \mz_m$.  

\begin{cor} The covering numbers of $Spin(m)\subset \mz_m\oplus \mz_m$ satisfy 
 \[ N(Spin(m),\eps) \kl N(O(m),\frac{\eps}{2})^2\kl (\frac{2C}{\eps})^{m(m-1)} \pl .\]
\end{cor} 
We will now deduce covering number estimates for `the inverse' of $o$. There is a group isomorphism  $\pi:O(m)\rtimes \zz_2\to {\rm Spin}(m)\subset \mz_{2^m}$ such that $o\pi(o,z)=o$. The Lie algebra $\mathfrak{so}(n)$ is the span of the elementary generator 
 \[ W_{jk} \lel |j\ran\lan k| - |k\ran\lan j| \pl. \]
Since $[\gamma_j\gamma_k,\gamma_k\gamma_l]=2\gamma_j\gamma_l$, it turns out that $u(2W_{jk})=\gamma_j\gamma_k$  is the desired Lie algebra homomorphism, i.e. $d\pi=u$.   As suggested  in Szarek \cite{Szarek1998MetricSpaces}, we may also consider the geodesic length with respect to the operator norm to transport covering numbers. This metric is given by  
 \[ d_{fins}(u,1) \lel \inf\{\int_0^t \|a(s)\|_{op} ds \pl ,\pl u\lel \int_0^t a(s) ds\} \pl .\] 
Here $a(s)$ is a tangent vector. We deduce from representation theory  that 
  \[ u \lel \int_0^t a(s) ds \pl \Rightarrow \pl \pi(u) \lel \int_0^t  d\pi(a(s)) ds  \pl. \] 
Recall that the Finsler metric and the ordinary operator norm distance are equivalent, see \cite{Szarek1998MetricSpaces}. 
This gives an upper bound for the operator norm distance:
\begin{lemma}\label{covspin}
 \begin{enumerate}
 \item[i)] $d_{op}(\pi(u),1)\kl 6\pi m d_{op}(u,1)$.
 \item[ii)] $N_{op}(Spin(m),\eps)\kl (\frac{4Cm}{\eps})^{m(m-1)}$.
 \item[iii)] Let $G_0$ be the Lie group of the Lie algebra generated by  $S_0=\mathcal{Z}(1)\cup \mathcal{X}(2)$. The covering number of $G_0$ satisfies  
 \[ N(G_0,\eps) \kl (\frac{40Cn}{\eps})^{2n(2n-1)} \pl .\] 
 \item[iv)] Let $\widetilde{G}_0$ 
 be the Lie group generated by $S_0\cup \{X_1\}$. Then 
  \[ N(\widetilde{G}_0,\eps) \kl (\frac{40Cn}{\eps})^{(2n+2)(2n+1)} \pl .\] 
 \end{enumerate}
\end{lemma}    

\begin{proof} We start with 
 \[ \|\pi(u)-1\|_{op} \kl \int_0^t \|d\pi(a(s))\| ds \pl .\] 
Therefore $\|d\pi(a)\|_{op} \kl 3m \|a\|_{op}$ will imply $i)$. Let $A=\sum_{kl} a_{kl} W_{kl}$ be a typical vector in the tangent space, an antisymmetric real matrix. Then
 \begin{align*}
  \|A\|_2^2 &=  \tr(A^*A) \lel -\tr(A^2) \lel - \sum_{kl,k'l'} a_{kl}a_{k'l'} \tr(W_{kl}W_{k'l'}) \\ 
 &=  -\sum_{kl,k'l'}  a_{kl}a_{k'l'}  2(\delta_{lk'}\delta_{kl'}-\delta_{ll'}\delta_{kk'}) \\
 &= 2 \sum_{lk} a_{kl}^2 - 2 \sum_{kl}a_{kl}a_{lk}  \lel 4 \|a\|_2^2 \pl .
 \end{align*}
Note that the $a_{kl}$ are real and $a_{kl}=-a_{lk}$ since $A$ is antisymmetric. Now we consider a real rank-one matrix $b$ with coefficients $b_{kl}\lel B_kC_l$, $B_k,C_l\in \rz$.  We deduce that 
 \begin{align*}
 \|\sum_{kl} b_{kl}\gamma_k\gamma_l\|
 &\le  \| \sum_k B_k\gamma_k\| \pl \|\sum_l C_l\gamma_l\|
 \kl (\sum_k|B_k|^2)^{1/2} (\sum_l |C_l|^2)^{1/2} \pl .
 \end{align*}
For arbitrary rank-one matrices $b_{kl}=B_kC_l$. we decompose $B_k$ and $C_l$ into real and imaginary part and deduce from an extremality argument that 
\[   \|\sum_{kl} b_{kl}\gamma_k\gamma_l\| \kl 4 \tr(|b|) \pl .\] 
This shows that $d\pi(A)=\frac{1}{2} \sum_{kj} a_{kj}\gamma_{k}\gamma_j$ satisfies 
 \[ \|d\pi(A)\|_{op} \kl \frac{4}{2}   \pl tr(|a|) \kl 2 \sqrt{m} \tr(|a|^2)^{1/2} \kl \sqrt{m} \|A\|_2
 \kl m \|A\|_{op}
   \pl .\]   
In particular,
 \[ d_{op}(\pi(u),1) \kl m  \pl d_{fins}(u,1) \kl 6\pi m  \|u-1\|_{op} \pl .\] 
Note that all these distances are right-invariant, i.e.\ $d(u,v)=d(uv^{-1},1)$; hence this estimate applies to covering numbers. 

For the proof  of the second assertion, we apply i) to $m=2n$ and use the Jordan-Wigner transform.  We first replace $\tilde{S}_0=S_0\cup \{X_1\}$ with $\hat{\tilde{S}}_0= S_0\cup \{X_1X_0\}$. According to Lemma \ref{lem:PauLie_homomorphism} and Lemma \ref{cstar}, the group and $C^*$-algebra are isometrically isomorphic. Therefore the covering numbers of their corresponding unit balls are identical. Then we note that $S_0(n)\cup \{X_1X_0\}$ is contained in the set $\{X_0X_1,X_1X_2,...,X_{n-1}X_n\}\cup \{Z_0,...,Z_n\}$ which is the set $S_0$ for $n+1$ elements. Thus Proposition \eqref{basic} ii) implies the assertion after applying i) for $m=2n+2$. 
\end{proof}

Let us now consider the specific Hamiltonian $H=\sum_j X_{j}X_{j+1}$ with the local resource set
 \[ S_0\lel \{X_1,Z_1Z_2\}\cup \mathcal{Z}(1).\] 
Here we need a slightly different argument:  
\begin{lemma}\label{prod} Let $G$ be the Lie group generated by $S_0$. Then 
 \[ N(G,\eps) \kl (\frac{4C}{\eps})^{2n+12} \pl .\] 
\end{lemma} 
\begin{proof} We observe that  
 \[ S_0\lel \{X_1,Z_1Z_2\}\cup \mathcal{Z}(1) \subset
  \{X_1,Z_{1}Z_2,Z_1,Z_2\} \cup \{Z_3,...,Z_{n}\} \pl \] 
This means that the group $G=G(S_0)\subset G_1G_2$ is contained in the product of the commuting groups $G_1=SU(4)$ and $G_2\mathbb{T}^{n-2}$. We can use the proof of Lemma \ref{basic} and combine it with 
 Theorem \ref{Szarek} to obtain $2(n-2)+16$. 
\end{proof}

\begin{lemma}\label{lem:covering_exp(tH)} Let H be a self-adjoint operator, $G_H =\{e^{itH}|t \in \rz \}$ and $\eps\le 1$. 
\begin{enumerate}
\item[i)] $N(G_H,\eps)\kl (\frac{C}{\eps})^{2^{n+1}}$;
\item[ii)] Let $k=rk(H)$ be the rank. Then $N(G_H,\eps)\le (1+\frac{Ck}{\eps})^k$ for some absolute constant $C=5\pi$.
\item[iii)] If H has integer spectrum, then $N(G_H,\eps)\kl 4\frac{2\pi +\|H\|}{\eps}$.
\end{enumerate}
\end{lemma} 

\begin{proof} For the proof of i) we diagonlize $H$ and then $e^{itH}$ belongs to the unit ball of $\ell_{\infty}^{2^n}$. For every complex Banach space $N(B_{X},\eps)\le (1+\frac{2}{\eps})^{2{\rm dim_{\cz} X}}$, see \cite{PSV}. For $dim_{\cz}\ell_{\infty}^{2^n}=2^n$, we 
deduce the assertion. For ii), let us first consider the orbit $G_H =\{e^{itH}|t\in \rz\}$ where $H=\la f$ is given by a single 
eigenvalue. Then $G_H = \{e^{itH}| 0\le t\le \frac{2\pi}{\la}\}$ because $e^{\frac{2\pi}{\la}H}=(1-f)+e^{2\pi i}f=1$ is the
corresponding period. We will use the bounds
 \[ 
 \|e^{itH}-e^{isH}\|\kl |t-s| \|H\|\kl |t-s||\la| \pl .
 \] 
Let us choose discrete time steps $t_j=\frac{j}{N}$ running from $j=1$ to $j=j_{\max}$ s.t. $j_{\max}-1\kl N\frac{2\pi}{|\la|}\kl j_{max}$. We find an $\eps=\frac{|\la|}{N}$ net of cardinality
 \[ j_{max}\kl N\frac{2\pi}{|\la|}+1 \lel \frac{2\pi}{\eps}+1 \pl .\] 
Adjusting for the correct $\eps$, we need another factor $2$. This also works for $\eps>|\la|$. To complete the proof of ii) we note that the Lie algebra   $\rz H\subset \rz f_1+\cdots +\rz f_k$ is contained in  commuting pieces. Then we use Proposition \ref{basic}iv). 

To prove the third  assertion iii), we note that $S=\{e^{itH}|0\le t\le 2\pi\}$. The net $e^{\frac{2\pi i j}{N}H}$ gives us an $\eps$-net for $\eps=\frac{2\pi\|H\|}{N}$ of cardinality $\frac{2\pi\|H\|}{\eps}+1$. 
\end{proof} 
We will frequently use the following elementary estimate: 
\begin{lemma}\label{ylny} Let $e\le x\kl K y\ln y$. Then
 \[ \frac{x}{K\ln x} \kl y \pl. \]
\end{lemma}
\begin{proof} Note that the function $g(y)=Ky\ln y$ is monotonic for $y\gl e^{-1}$. And for $y\le e^{-1}$ we see that $g(y)<0$ contradicting the assumption.    
Let $x\gl e$, then 
 \[ g(\frac{x}{K \ln x})\lel \frac{Kx}{K \ln x} (\ln x-
 \ln \ln x)\kl x \kl g(y) \]
implies $\frac{x}{K\ln x}\le y$.     
\end{proof}

\begin{theorem}\label{entropy} Let $n\gl 2$. Let $c$ and $C$ be given by Szarek's theorem bounding the covering numbers for $U(2^n)$, see Theorem \ref{Szarek}. Let $G_0\subset U(2^n)$ be a Lie group of dimension $d$ with covering numbers
 \[ N(G_0,\delta) \kl (\frac{Ca}{\delta})^d, \]
and $H$ be a magic Hamiltonian. Let $\eps<\frac{1}{2c}$ with $c$ given by Szarek's theorem (for all $U(2^n)$). Then   
 \begin{enumerate}
 \item[i)] If $H$ has integer spectrum, then there exists a constant $C_1(\eps)$ such that 
  \[ \frac{4^n}{n} \kl C_1(\eps) ~\max \{\ln (1+\|H\|),d(1+\ln a)\}~ l_H(\eps) \pl .\] 
 \item[ii)] If $H$ has rank $k$, then 
   \[ \frac{4^n}{n} \kl C_2(\eps) ~\max \{k(1+\ln k),d(1+\ln a) \}~ l_H(\eps) \pl .\] 
 \item[iii)] For arbitrary $H$ 
 \[ 4^{n}  \kl C_3(\eps)(d(n+ \ln a) + 2^nn) l_H(\eps) \pl. \]
\end{enumerate} 
Here $C_j(\eps)$ are absolute constants. 
\end{theorem}

\begin{proof} Let $\eps<\frac{1}{2c}$, where $c$ is the constant of Szarek's entropy result. Then we can find $m=m(\frac{\eps}{2})$ such that for all $u \in U(2^n)$, there exist $u_{2j+1} \in G_0$ and $u_{2j} \in \{ e^{it H} ~:~ t \in \R \}$ such that
\[ 
    \|u-u_1\cdots u_{2m+1}\|<\frac{\eps}{2} \pl
\] 
holds.
For every $u_{2j+1}$ we can find some $\hat{u}_{2j+1}\in S$ in a set of cardinality $|S|\le N(G_0,\frac{\eps}{2m})$ such that 
 \[ \|u_{2j+1}-\hat{u}_{2j+1}\|\kl \frac{\eps}{2(2m+1)} \pl .\] 
Similarly, we can find a subset $\tilde{S}\subset \rz$ such that $\|u_{2j}-e^{isH}\|<\frac{\eps}{2(2m+1)}$ for $s\in \tilde{S}$ and $|\tilde{S}|\le N(G_H,\frac{\eps}{2(2m+1)})$.  
Since all of these are unitaries, we deduce that 
 \[ \|u-\hat{u}_1\cdots \hat{u}_{2m+1}\|< \frac{\eps}{2} \sum_{j=1}^{2m+1}\|u_{j}-\hat{u}_j\| \kl \eps \pl .\] 
 This implies
\[ 
    \Big(\frac{c}{\eps} \Big)^{2^{2n}-1} \leq N(U(2^n),\eps)\kl N\Big(G_0,\frac{\eps}{2(2m+1)} \Big)^{m+1}  N \Big(G_H,\frac{\eps}{2(2m+1)} \Big)^{m}.
\]
Now we take the logarithm and deduce that  
 \begin{align*}
     (2^{2n}-1)\ln \frac{c}{\eps} & \kl (m+1)\ln  N\big(G_0,\frac{\eps}{2(2m+1)} \big)+m \ln  N\big(G_H,\frac{\eps}{2(2m+1)} \big) \\
     &\le 
     (2m+1) \max \Big\{\ \ln  N \big(G_0,\frac{\eps}{2(2m+1)} \big),~\ln   N \big(G_H,\frac{\eps}{2(2m+1)} \big) \Big\} \\
     &\le (2m+1) \max \Big\{ d \ln  \frac{Ca}{\eps} + d \ln (2m+1),~ \ln N \big(G_H,\frac{\eps}{2(2m+1)} \big) \Big\}. 
     \end{align*}
Now, we use the different assumptions for the estimate of $N(G_H)$. For arbitrary $H$, we have by Lemma \ref{lem:covering_exp(tH)}iii) that  
\[ 
    \ln N(G_H,\frac{\eps}{2(m+1)}) \kl 2^{n+1} \big(\ln \frac{C}{\eps}+\ln (2m+1) \big)  \pl .
\]
This leaves us with four cases:
  \begin{enumerate}
      \item[i)] $(2^{2n}-1)\ln \frac{c}{\eps} \kl  4 d(2m+1)\ln \frac{Ca}{\eps}$;
      \item[ii)] $(2^{2n}-1)\ln \frac{c}{\eps}\kl 4 d(2m+1)\ln (2m+1)$;
      \item[iii)]  $(2^{2n}-1)\ln \frac{c}{\eps}\kl 4 \cdot  2^{n+1} (2m+1) \ln \frac{C}{\eps}$;
      \item[iv)] $(2^{2n}-1)\ln \frac{c}{\eps}\kl 4  \cdot 2^{n+1} (2m+1)\ln (2m+1)$.
      \end{enumerate}
We use Lemma \ref{ylny} and deduce the assertion.  If $H$ has rank $k$, we have
 \[ \ln N(G_H,\frac{\eps}{2(m+1)}) \kl  k \ln (1+\frac{Ck(2m+1)}{\eps}) 
 \kl k (1+\ln \frac{Ck}{\eps})(1+\ln (2m+1)) \pl. \]
This means the cases iii) and iv) are replaced by better estimates and yield ii). If we have integer spectrum, we get 
 \[ \ln N(G_H)\kl  \ln 4+ \ln \frac{1}{\eps}+\ln (2\pi +\|H\|) \pl .\]
 Replacing again iii) and iv) by better alternatives, we deduce the assertion. \end{proof}

In the next result we show that the support function $l_H^0$ can be replaced by other functionals.

 \begin{prop}\label{mp} Either $l_H(\frac{\eps}{2p})=0$, or 
 \[ l_H(\eps)^{\frac{1-p}{p}} \kl \frac{p}{1-p} \frac{Cl_H^p(\eps/2)\|H\|}{\eps}  \pl, \]
 holds for an absolute constant $C$.   
\end{prop} 

\begin{proof} Let us assume that for every $u$ we find an approximation as above with $m_p(\eps/4)\kl M$. Let $t_j^*$ be the largest $t_j$ that occurs. Let us choose a subset $I$ such that  
 \[ \sum_{j\notin I} |t_j|<\frac{\eps}{2\|H\|}. \] 
By replacing all the $e^{t_jH}$ with $j\notin I$ by $1$, we find
 \[ \|u_1\cdots u_{2m+1}-\hat{u_1}\cdots \hat{u}_{2m+1}\|
 \kl \sum_{j\in I} \|u_{2j}-\hat{u}_{2j}\| \kl \frac{\eps}{2} \pl .\] 
Since we are looking at an alternating product, for $j\notin I$ we see that $u_{2j-1}u_{2j+1}\in G$. Therefore, we effectively have only $j\in I$ alternating elements. 
 \[ l_{H}(\eps) \kl \sup_u |I(u)| \pl .\] 
Now assume that $p<1$ and $l_{H}^p(\eps/2)<
\infty$. Then
 \[ j^{1/p}t_j^* \kl \|t\|_p \kl l_{H}^p(\eps/2)  \] 
holds for the non-increasing rearrangement. Then note that 
 \[ \sum_{j\gl k} t_j^* \kl   l_{H}^p(\eps/2) \sum_{j>k} j^{-1/p} 
 \kl \frac{C}{\frac{1}{p}-1}  l_{H}^p(\eps/2) k^{1-1/p} \pl .\] 
We deduce that for $k^{1/p-1}(1/p-1)\gl \frac{2\|H\| l_H^p(\frac{\eps}{2})}{\eps}\gl (k-1)^{1/p-1}(1/p-1)$ we have
 \[ l_H(\eps) \kl k-1 \kl  \Big(\frac{p}{1-p}\Big)^{\frac{p}{1-p}} \bigg(\frac{2\|H\|l_H^p(\eps/2)}{\eps}\bigg)^{\frac{p}{1-p}}  \pl . \]    
Taking this to the power $\frac{1-p}{p}$ yields the assertion. If we find no such $k$, then $\frac{2\|H\|l_H^p(\frac{\eps}{2})}{\eps}\le 1/p-1$. This means
 \[ l_H^1(\eps/2) \kl \frac{\eps}{2\|H\|}(1/p-1), \]
and we can approximate every $u$ with an element in $G_0$ of distance $\frac{\eps}{2p}$.  
\end{proof}  

\begin{rem}
\begin{enumerate}
\item[i)] The motivation to consider $l_H^1(\eps)$ is the extensive use of the Trotter formula in the proof of Proposition \ref{finite}. The Trotterization does not change the $l^1$-norm but dramatically increases the support. At the time of this writing, no lower bound for alternating functional $l_H^1(\eps)$ is known and this remains an open problem. 
\item[ii)] The proof above also shows that the number of essential peaks
 \[ l_H^*(\eps) \lel \inf \Big\{|I| | \sum_{j\notin I} |t_j|\kl \frac{\eps}{2\|H\|}  \Big\} 
 \] 
satisfies $l_H(\eps)\kl l_H^*(\eps)$ and hence also is usually exponential. 
\end{enumerate}   
\end{rem} 

\begin{theorem}\label{thm:exponential_m} For all magic Hamiltonians in Theorem \ref{themagic}, namely cases 1), 2), 4), 5), 7), 8), 9), the alternating length is exponential. 
\end{theorem}

\begin{proof} For 1), 2) we use the Jordan-Wigner transform; more precisely, Lemma \ref{covspin} iii) for $m=2n$ or $m=2(n+1)$. This leads to a constant $n\ln n$ coming from local  Lie group $G_0$ generated by $S_0$. On the other hand, the Hamiltonian can be considered in $\mz_4$ embedding isometrically into $\mz_{2^n}$. Therefore, we obtain a lower bound $l_H(\eps)\gl c(\eps) \frac{4^n}{n^2\ln n}$. 

For 4) and 5) we have a very small local Lie group $G_0$, either contained in $\mathbb{T}^n$ or $SU(4)\times \mathbb{T}^{n-2}$. The Hamiltonian has integer spectrum and $\|H\|\le n$, hence $l_H(\eps)\gl c_{\eps}\frac{4^n}{n}$. 

In 8) and 9) we have again small local Lie algebras. Indeed, we can work with the Lie algebra $g_{0,1}=\langle \mathfrak{su}(2),Z_1Z_2\rangle$ which commutes with $g_{0,2}=\langle Z_{2}Z_3,...,Z_{n-1}Z_n\rangle$. Therefore the Lie group is contained in the product 
 \[ G_{0,1}\times G_{0,2} \subset SU(4)\times \mathbb{T}^{\lfloor n/2 \rfloor} \pl. \]
This will give a term of order $n$. However, since we have no additional benefit for the Hamiltonian, we deduce a lower bound $c_{\eps}\frac{2^n}{n}$. 

In number 7) we split the generators into two commuting sets
\[  S_{lloc}=\{X_1,Z_1,X_1X_2X_3X_4, Z_1Z_2Z_3Z_4,Z_1Z_2,Z_2Z_3,Z_4Z_5\}  , 
T=\{Z_5Z_6,Z_6Z_7,...,Z_{n-1}Z_n\}\pl .\] 
In our PauLie map we can replace $Z_5Z_6$ by $Z_6$. In the new picture, we find a product group.  $SU(2^5)\times \mathbb{T}^{n-6}$ and ransport the covering number estimates back. Therefore we can combine Proposition \ref{covspin}ii) and iv) with Szareks results Theorem  \ref{Szarek} and find 
 \[ N(G,\eps)\kl (\frac{C}{\eps})^{2^{5+1}+2(n-6)} \pl. \]
We use Lemma \ref{basic} and obtain $a=1$ and $d=54+2n$. Again, we have to accept the general Hamiltonian estimate and get $c(\eps)\frac{2^n}{n}$ as a lower bound. 
\end{proof}

\begin{rem} The case 5) with $H=\sum_j X_jX_{j+1}$, i.e. $\gamma=1$, $\la=0$ is still accessible with our methods. Indeed for even $n$ we generate a copy of $G=SU(2^{n/2})\times SU(2^{n/2})$. Szarek's argument is easily seen to remain valid for a product of two groups. Therefore, we can use 
 \[ (\frac{c}{\eps})^{2(4^{n/2}-1)} \kl N(G,\eps) \]
in the proof of Theorem \ref{entropy}. The good news is that we have $a=1$ and $d=n/2$ and $\|H\|=n$ for our local estimate and Hamiltonian. Thus, we get a lower bound of $l_H(\eps)\gl c(\eps) \frac{2^n}{n}$ to approximate every unitary in the smaller group $G$ (not $SU(2^n)$). We note that the difference between 5) and 6) corresponds exactly to the predicted phase transition at $\gamma=1$ fixed  and  $\la=0$ versus $\la\neq 0$ from operator algebra theory, see Appendix \ref{app:operator_algebra}. 
On the other hand 7) and 8) correspond to phase transitions in both parameters.
\end{rem}

\begin{rem} It remains completely open whether these lower bounds are sharp. An upper bound for $e^{tV}$, with $V$ is a Pauli string, can be obtained by applying our algorithmic formula for iterated commutators together with Rashevskii's technique. Even for small $t^{1/k}$, the number of terms in these formula grows as $l_{j+1}=2+2l_j$ and hence at least exponentially ($2^k$) in the number $k$ of iterated commutators. This can then be combined with Trotterization, see \ref{finite}.  This method will require at least $2^{n^3}$ many H\"ormander generators and the algorithm will tell us how many of them are from the magic Hamiltonian.      
\end{rem}

\section{Conclusion and Discussion}\label{sec:conclusion}
In this paper we have shown that Hamiltonian magic is possible: A small set of local generators can be combined with well-known Hamiltonians such as the Ising model to generate the full Lie algebra, and ultimately obtain control of the full unitary group. Both procedures are algorithmic, and in our example the Lie algebra generation is of order poly$(n)$. 
As discussed in Appendix \ref{app:operator_algebra}, these choices of parameters for the Ising model miraculously coincide with the phase transitions for the number of extreme points in the ground state for the string-net model.  However, full control over the unitary group may not be cheap, as shown by the exponential lower bounds for the alternating length.  

What does this imply for experimental quantum platforms? Let's compare our magic Hamiltonian-based regime for computation with the standard model of quantum computation. In the standard regime, one has sufficient power (perhaps as a result of successful error-correction schemes) to already accurately address each individual qubit without recourse to a magic Hamiltonian. It seems clear that the natural way to create a desired unitary here is just to use the canonical universal gate set approach. We remind the reader, however, that in physical systems these gates are in fact attained by evolution with a full Hamiltonian, which in practice contains additional undesirable interactions. Our results show how even small pulses can be spread by this additional interaction to create errors across the entire system. It would be interesting to explicitly reformulate our results in the language of noise channels and error propagation.

In our alternative magic Hamiltonian-based regime, one instead only has full control over a small number of qubits in an experimental system, plus some more limited ability to directly affect the remainder of the system, for example only by Pauli $Z$ operations which all commute. What our results show is that, with the added power of a magic Hamiltonian, one can still do arbitrary computations on the full system even without the ability to directly control most of the qubits in it. The cost is of course that the time (number of Hamiltonian applications) to create a unitary will typically be exponential. Nevertheless, if one was previously under the impression that computations on these systems were not possible at all, the ability to do computations, and implement long-range entanglement, in exponential time is still interesting.
More generally, one could imagine a full spectrum of tradeoffs between the amount of control over the system and the time needed to perform a computation, and depending on the application it might not always be desirable to live at an extreme.

\bibliographystyle{unsrt}
\bibliography{References.bib}

\appendix
\section{Explicit mapping from \texorpdfstring{$S_i$}{S} to \texorpdfstring{$T_i$}{T}}\label{app:explicit_map_S_T}
In Prop.~\ref{prop:double-double}, we use the properties of PauLie maps to find the resulting Lie Algebra without needing an explicit map on all Pauli strings. Here we do give the explicit map mapping the generators in $S_1$ to $T_1$ and in $S_2$ to $T_2$, given by
\begin{equation}
\begin{aligned}
    S_1 &= \{X_1 X_2, Y_2 Y_3, \cdots Y_{n-2}Y_{n-1},  X_{n-1} X_n\} \\
    S_2 &= \{Y_1 Y_2, X_2 X_3, \cdots X_{n-2}X_{n-1},  Y_{n-1} Y_n\} \\
    T_1 &= \{X_1, Y_1 Y_2, X_2, \cdots Y_{n/2-1} Y_{n/2} X_{n/2} \} \\
    T_2 &= \{Y_{n/2+1}, X_{n/2+1}X_{n/2+2}, Y_{n/2+2}, \cdots  X_{n-1} X_n, Y_n \}.
\end{aligned}
\end{equation}
One can extend this map to all Pauli strings while preserving the commutation relations. The explicit map for the Pauli strings with one non-identity component is given by
\begin{equation}
\begin{aligned}
    \sigma(X_k) &= \begin{cases}
        X_1 X_2 \cdots X_{k/2} \otimes X_{n/2 + k/2} & \text{if $k$ even} \\
        X_{(k+1)/2}\cdots X_{n/2} \otimes X_{n/2 + (k+1)/2} & \text{if $k$ odd}
    \end{cases}\\
    \sigma(Y_k) &= \begin{cases}
        Y_{k/2} \otimes Y_{n/2+1} Y_{n/2+2} \cdots Y_{n/2+ k/2} & \text{if $k$ even} \\
        Y_{(k+1)/2} \otimes Y_{n/2+ (k+1)/2} \cdots Y_{n} & \text{if $k$ odd}
    \end{cases} \\
    \sigma(Z_k) &= \begin{cases}
        X_1 X_2 \cdots X_{k/2-1} Z_{k/2}  \otimes Y_{n/2+1} Y_{n/2+2} \cdots Y_{n/2+ k/2-1} Z_{n/2+k/2} & \text{if $k$ even} \\
        Z_{(k+1)/2} X_{(k+1)/2+1}\cdots X_{n/2} \otimes Z_{n/2+(k+1)/2} Y_{n/2+ (k+1)/2+1} \cdots Y_{n} & \text{if $k$ odd},
    \end{cases} 
\end{aligned}
\end{equation}
where the $\otimes$ is added between the two subgroups for clarity.
This map preserves the commutation relations between the single operators $X_k, Y_k$, and $Z_k$ and the operators in $S_1$ and $S_2$. Since $\cX, \cY, \cZ, \cX(2),\cY(2)$ together generate the full Lie algebra, one can find the action of the map on any Pauli string by writing it as a nested commutators of these generators, applying $\sigma$ to all those generators, and then working out all the commutators.

\section{Pseudo-algorithm}\label{app:alg-computation}

For completeness and clarity, we append here in Alg.~\ref{alg:implementation} a pseudo-code version in Alg.~\ref{alg:implementation} of the explicit algorithm to find the nested commutators of a given Pauli string according to Thm.~\ref{thm:commutators_algorithmically}.
\begin{algorithm}[htbp]\label{alg:implementation}
\SetAlgoLined
\SetInd{0.5em}{0.5em}
\SetKwData{Left}{left}\SetKwData{This}{this}\SetKwData{Up}{up}
\SetKwFunction{Union}{Union}\SetKwFunction{FindCompress}{FindCompress}
\SetKwInOut{Input}{input}\SetKwInOut{Output}{output}
\Input{$\alpha_1 \alpha_2 \cdots \alpha_n$, $\alpha_1 \neq I$}
\Output{List of generator elements $V$ to commute with in sequence.}
\BlankLine
$V = [Z_1]$     \Comment{Start with $Z_1$} \\
$k \gets n$ \\
\While{$k>2$}{
\tcp{Make $Z_1 I\cdots I \alpha_k \alpha_{k+1} \cdots \alpha_n$}
\If{$\alpha_k=Z$}{
    $V \gets [X_1, Z_1, Z_1Z_2, X_1] +V$  \Comment{Makes $Z_1Z_2$}\\
    \For{$l=2$ \KwTo $k-1$}{
    $V \gets [X_l X_{l+1}, Z_l, Z_{l+1}, X_l X_{l+1}]+V$ \Comment{Makes $Z_1Z_k$}
    }
}
\ElseIf{$\alpha_k \in \{X, Y\}$}{
    \For{$l=2$ \KwTo $k$}{ 
    $V \gets [X_1, Z_1, Z_1Z_2, X_1] + V$  \\
    \For{$m=2$ \KwTo $l-1$}{ 
        $V \gets [X_l X_{l+1}, Z_l, Z_{l+1}, X_l X_{l+1}] + V$ \Comment{Makes $Z_1 Z_2 \cdots Z_k$}
        }
    }
    $V \gets [X_1 X_2, Z_1 Z_2, Z_1, X_1]+V$ \\
    \For{$l=2$ \KwTo $k-2$}{
        $V \gets [Z_{l+1}, X_l X_{l+1}] + V$
    }
    $V \gets [X_{k-1} X_k] + V$  \Comment{Makes $Z_1 Y_k$}\\
    \If{$\alpha_k=X$}{
        $V \gets [Z_k]+V$ \Comment{Makes $Z_1 X_k$}
    }
}
$k \gets k-1$
}
Append appropriate generators $S$ to make $\alpha_1 \alpha_2$ from $Z_1 I_2$.\\
$V \gets S+V$\\
return $V$;
\caption{Explicit sequence of commutators}\label{algo:explicit_commutators}
\end{algorithm}

\section{Motivation from operator algebra theory}\label{app:operator_algebra}

Phase transitions for string-net models provide profound insights for quantum computation and the ability to simulate large number of qubits. The theory of quasi local algebras provides very precise formulations on what phase transitions can mean in infinite dimension see in particular \cite{ArEw,EvLew}. Our focus on the $XY$ model is particularly inspired by the results of Araki and Matsui. A quasi local algebra is the limit of local algebras
 \[ A_{\La} \lel \prod_{p\in I} A_p \pl ,\pl \mathcal{A} \lel \lim_{\La\to \infty} A_{\La} \]
where $A_p$ is a collection of finite dimensional matrix algebras indexed by an integer lattice. The quasi-local algebra $\A$ is the natural direct limit of the local algebras. In the $C^*$-sense phase transition are expressed as properties of irreducible representations of $\mathcal{A}$. Of particular interest are so-called GNS representation $A_{\phi}=\pi_{\phi}(\mathcal{A})''$ given by a state $\phi:\mathcal{A}\to \mathbb{B}(H_{\phi})$. Here $H_{\phi}$ is the Hilbert space obtained from the inner product
 \[ (x,y)_{\phi} \lel \phi(x^*y) \pl .\]
$H_{\phi}$ admits a cyclic vector $\xi_{\phi}\cong 1$ such that $(\phi(x)=(\xi_{\phi},x\xi_{\phi})$. Associated with a given Hamiltonian different states have been constructed. 
An interaction Hamiltonian 
 \[ H \lel \sum_{J\subset \zz^d} H_z  \]
can usually be approximated by the local Hamiltonian   
 \[ H_{\La} \lel \sum_{J\subset \La} H_{J} \pl.  \]
One possibility is to take a cluster point of the corresponding Gibbs states
 \[ \phi_{\beta,\La}(x) \lel \frac{tr(e^{-\beta H_{\La}}x)}{tr(e^{-\beta H_{\La}}x)} \pl. \] 
Under mild conditions \cite{BR1,BR2} on the interaction Hamiltonian such a cluster point $\phi$ still satisfies the KMS condition 
\[ \phi(\al_t(x)y) \lel \phi(y\al_{-i\beta+t}(x)) \]
for the automorphism group 
 \[ \al_t(x) \lel \lim_{\La \to \infty}  e^{iH_{\la}}xe^{-itH_{\La}} \pl . \]
It is important and non-trivial to show that $\al_t$ extends to an automorphism of the closure $\mathcal{A}$.   
In \cite{ArakiB} ergodic properties, and more precisely the return to equilibrium is studied for the KMS state of the $XY$ model 
 \[ J_{\gamma,\la} \lel \sum_{j} (1+\gamma)X_{j}X_{j+1}+(1-\gamma)Y_jy_{j+1}+ \la \sum_j Z_j\]
on the one dimensional lattice. The result by Araki and Matsui \cite{AM} for the ground states are even more explicit. For this we have to define the derivation 
 \[ \delta(x) \lel i[H,x] \]
which is well-defined for elements in $A_N=A_{[-N.N]}$ and every $N$. Clearly this is an interaction model with uniformly bounded bounded terms because we are in 1D. A state is called ground state if
 \[ i\phi(x^*\delta(x))\gl 0 \]
for all $x\in \cup_N A_N$. 
\begin{theorem}[Araki-Matsui] The number of extremal ground states is
 \begin{enumerate}
     \item[i)] $1$ if $\gamma=0$ and $|\la|<2$;
     \item[ii)] $2$ if $\gamma\neq 0$ and $(\la,\gamma)\neq (0,\pm 1)$;
     \item[iii)] $\infty$ if $(\gamma,\lambda)=(\pm 1,0)$.
 \end{enumerate}
\end{theorem}
The Jordan-Wigner transform, also used in section 2, is a crucial tool in this result together with the notion of quasi-equivalence corresponding to central projections of $\mathcal{A}^{**}$ supporting GNS representations. Certainly, the case ii) could indicated different phases for ground states, in the $C^*$-sense. In our computational phase transitions the role values $\gamma=\{-1,0,1\}$ give rise to `computational phase transitions'
and the same applies to $\la =0$, depending of our choice of local resource.

\end{document}